\documentclass{article}

\usepackage[boxed,nofillcomment,linesnumbered,noend,noresetcount, noline]{algorithm2e}
\usepackage{times}
\usepackage{fullpage}
\usepackage{cite}
\usepackage{graphicx}
\usepackage{epsf}
\usepackage{tabularx}
\usepackage{multirow}
\usepackage[small, center]{caption}
\usepackage{amsmath}
\usepackage{amsthm}
\usepackage{amssymb,latexsym,color}
\usepackage{mdwlist}

\newtheorem{lemma}{Lemma}
\newtheorem{theorem}{Theorem}

\newtheorem{corollary}{Corollary}

\newtheorem{fact}{Fact}

\newenvironment{proofT}[1]{\proof\def\toto{#1}}{\hspace*{\fill}$\Box_{Theorem~\ref{\toto}}$\par\vspace{3mm}}

\newenvironment{proofL}[1]{\proof\def\toto{#1}}{\hspace*{\fill}$\Box_{Lemma~\ref{\toto}}$\par\vspace{3mm}}
\newenvironment{proofC}{\proof}{\hspace*{\fill}$\Box_{Corollary~\ref{\toto}}$\par\vspace{3mm}}

\newcommand{\thmpostponed}[2]
{
\newcounter{#1}
\setcounter{#1}{\value{theorem}}
\begin{theorem}
\label{thm:#1}
#2
\end{theorem}
\expandafter\def\csname #1\endcsname{
\newcounter{#1temp}
\setcounter{#1temp}{\value{theorem}}
\setcounter{theorem}{\value{#1}}
\begin{theorem}
#2
\end{theorem}
\setcounter{theorem}{\value{#1temp}}
}
\vspace{-1.5em}
}

\newcommand{\thmpostponedwname}[3]
{
\newcounter{#1}
\setcounter{#1}{\value{theorem}}
\begin{theorem}[#2]
\label{thm:#1}
#3
\end{theorem}
\expandafter\def\csname #1\endcsname{
\newcounter{#1temp}
\setcounter{#1temp}{\value{theorem}}
\setcounter{theorem}{\value{#1}}
\begin{theorem}[#2]
#3
\end{theorem}
\setcounter{theorem}{\value{#1temp}}
}
\vspace{-1.5em}
}

\newcommand{\lemmaproofpostponedwname}[4]
{
\newcounter{#1}
\newcounter{#1temp}
\setcounter{#1}{\value{lemma}}
\begin{lemma}[#2]
\label{#1}
#3
\end{lemma}
\expandafter\def\csname #1\endcsname{
\setcounter{#1temp}{\value{lemma}}
\setcounter{lemma}{\value{#1}}
\begin{lemma}
#3
\end{lemma}
\setcounter{theorem}{\value{#1temp}}
\begin{proofL}{#1}
#4
\end{proofL}
}
}

\newcommand{\thmproofpostponedwname}[4]
{
\newcounter{#1}
\newcounter{#1temp}
\setcounter{#1}{\value{theorem}}
\begin{theorem}[#2]
\label{#1}
#3
\end{theorem}
\expandafter\def\csname #1\endcsname{
\setcounter{#1temp}{\value{theorem}}
\setcounter{theorem}{\value{#1}}
\begin{theorem}
#3
\end{theorem}
\setcounter{theorem}{\value{#1temp}}
\begin{proofT}{#1}
#4
\end{proofT}
}
}

\newcommand{\lemmapostponed}[2]
{
\newcounter{#1}
\newcounter{#1temp}
\setcounter{#1}{\value{theorem}}
\begin{lemma}
\label{#1}
#2
\end{lemma}
\expandafter\def\csname #1\endcsname{
\setcounter{#1temp}{\value{theorem}}
\setcounter{theorem}{\value{#1}}
\begin{lemma}
#2
\end{lemma}
\setcounter{theorem}{\value{#1temp}}
}
\vspace{-1.5em}
}

\newcommand{\lemmapostponedwname}[3]
{
\newcounter{#1}
\newcounter{#1temp}
\setcounter{#1}{\value{theorem}}
\begin{lemma}[#2]
\label{#1}
#3
\end{lemma}
\expandafter\def\csname #1\endcsname{
\setcounter{#1temp}{\value{theorem}}
\setcounter{theorem}{\value{#1}}
\begin{lemma}[#2]
#3
\end{lemma}
\setcounter{theorem}{\value{#1temp}}
}
\vspace{-1.5em}
}

\newcommand{\lemmaproofpostponed}[3]
{
\newcounter{#1}
\newcounter{#1temp}
\setcounter{#1}{\value{theorem}}
\begin{lemma}
\label{#1}
#2
\end{lemma}
\expandafter\def\csname #1\endcsname{
\setcounter{#1temp}{\value{theorem}}
\setcounter{theorem}{\value{#1}}
\begin{lemma}
#2
\end{lemma}
\setcounter{theorem}{\value{#1temp}}
\begin{proofL}{#1}
#3
\end{proofL}
}
}

\newcommand{\corollaryproofpostponed}[3]
{
\newcounter{#1}
\newcounter{#1temp}
\setcounter{#1}{\value{theorem}}
\begin{corollary}
\label{#1}
#2
\end{corollary}
\expandafter\def\csname #1\endcsname{
\setcounter{#1temp}{\value{theorem}}
\setcounter{theorem}{\value{#1}}
\begin{corollary}
#2
\end{corollary}
\setcounter{theorem}{\value{#1temp}}
\begin{proofC}
#3
\renewcommand{\toto}{#1}
\end{proofC}
}
}

\renewcommand{\paragraph}[1]{\vspace{0.1cm} \noindent\textbf{#1}\hspace{0.15em}}

\newcommand{\toto}{xxx}

\SetFuncSty{textsf}
\SetFuncSty{small}

\newtheorem{definition}{Definition} 

\usepackage{algorithm2e}

\newcommand{\E}{\mathop{\mathbb{E}}}

\newcommand{\Exp}{\mathrm{Exp}}
\newcommand{\rank}{\mathrm{rank}}
\newcommand{\Poi}{\mathrm{Poi}}
\newcommand{\poly}{\mathrm{poly}}

\newcommand{\del}{\textsc{deleteMin}}

\newtheorem*{rep@theorem}{\rep@title}
\newcommand{\newreptheorem}[2]{%
\newenvironment{rep#1}[1]{%
 \def\rep@title{#2 \ref{##1}}%
 \begin{rep@theorem}}%
 {\end{rep@theorem}}}
\makeatother

\newreptheorem{theorem}{Theorem}
\newreptheorem{lemma}{Lemma}
\newreptheorem{claim}{Claim}
\newreptheorem{invariant}{Invariant}
\newreptheorem{corollary}{Corollary}

\usepackage{hyperref}
\hypersetup{
    unicode=false,          
    colorlinks=true,        
    linkcolor=red,          
    citecolor=green,        
    filecolor=magenta,      
    urlcolor=cyan           
}
\usepackage[capitalize]{cleveref}

\DeclareMathOperator*{\maxrank}{\mathop{\max \mathrm{rank}}}

\begin{document}

\title{The Power of Choice in Priority Scheduling}

\author{Dan Alistarh\\IST Austria \& ETH Zurich\\\texttt{dan.alistarh@inf.ethz.ch}
\and
Justin Kopinsky\\Massachusetts Institute of Technology\\\texttt{jkopin@mit.edu}
\and
Jerry Li\\Massachusetts Institute of Technology\\\texttt{jerryzli@mit.edu}
\and
Giorgi Nadiradze\\ETH Zurich\\\texttt{giorgi.nadiradze@inf.ethz.ch}}

\maketitle

\date{}

\begin{abstract}
Consider the following random process: we are given $n$ queues, into which elements of increasing labels are inserted uniformly at random. To remove an element, we pick two queues at random, and remove the element of lower label (higher priority) among the two. The \emph{cost} of a removal is the \emph{rank} of the  label removed, among labels still present in any of the queues, that is, the distance from the optimal choice at each step. Variants of this strategy are prevalent in state-of-the-art concurrent priority queue implementations. Nonetheless, it is not known whether such implementations provide any rank guarantees, even in a sequential model.

 We answer this question, showing that this strategy provides surprisingly strong guarantees: Although the single-choice process, where we always insert and remove from a single randomly chosen queue, has degrading cost, going to infinity as we increase the number of steps, in the two choice process, the expected rank of a removed element is $O( n )$ while the expected worst-case cost is $O( n \log n )$. These bounds are tight, and hold irrespective of the number of steps for which we run the process.   
 The argument is based on a new technical connection between ``heavily loaded" balls-into-bins processes and priority scheduling. 
 Our analytic results inspire a new concurrent priority queue implementation, which improves upon the state of the art in terms of practical performance. 
\end{abstract}

%
%
 

\section{Introduction} 
\setcounter{page}{1}

The last two decades have seen tremendous progress in the area of concurrent data structures, to the point where scalable versions of many classic data structures are now known, including lists~\cite{Harris, LeapList}, skip-lists~\cite{skiplist2, LazySkiplist, DougSkiplist}, hash-tables~\cite{Michael, hopscotch,java-hash-map,shalev-shavit}, or search trees~\cite{CBTree, Lewandoski, Brown14}.  
In this context, significant research attention has been given to concurrent \emph{producer-consumer} data structures, such as queues, stacks, or priority queues, which are critical in concurrent systems. 
On the one hand, several elegant designs, e.g.~\cite{Morrison13, Fatourou12, Fatourou14, Morrison16}, have pushed the these data structures to the limits of scalability on current processor architectures. 
On the other, impossibility results~\cite{Alistarh14, EllenHS12} show that such data structures providing strong ordering semantics also have inherent sequential bottlenecks which limit their scalability.

A popular approach for circumventing these impossibility results is to \emph{relax} data structure semantics~\cite{ShavitCACM}. 
The first instance of this strategy is probably the relaxed PRAM priority queue structure of Karp and Zhang~\cite{KarZha93}, designed for efficient parallel execution of branch-and-bound programs. 
A flurry of subsequent work, e.g.~\cite{San98, Basin11, LotanShavit, linden2013skiplist, klsm, Nguyen13, Haas, SprayList, MQ}, has shown that relaxed semantics can yield extremely good practical results, and that it can be very useful in the context of higher-level constructs, such as concurrent graph algorithms~\cite{Nguyen13}.\footnote{In practice, it is often possible to offset the cost of \emph{priority inversions} by performing \emph{additional work}. For instance, in Dijkstra's single-source shortest-paths algorithm, priority inversions can be offset by nodes being relaxed multiple times.} 

Notably, the current state-of-the-art concurrent priority queue is obtained by the following 
\emph{MultiQueue} strategy~\cite{MQ, Wimmer}. We start from $n$ queues, each protected by a lock. 
To \emph{insert} an element, a processor picks one of the $n$ queues uniformly at random, locks it, and inserts into it. 
To \emph{remove} an element, the processor picks \emph{two queues} uniformly at random, locks them, and removes the element of \emph{higher priority} among their two top elements. 
 Extensive testing~\cite{MQ, Wimmer} has shown that this natural strategy can provide state-of-the-art throughput, and that the average number of priority inversions is relatively low.  
Further variants of this strategy are used to implement general priority schedulers~\cite{Nguyen13} and relaxed concurrent queues~\cite{Haas}. 

Despite their significant practical success, randomized relaxed concurrent data structures still lack a theoretical foundation: it is not known whether distributed strategies such as the MultiQueue can provide any \emph{guarantees} in terms of their inversion cost, or whether their performance can still be improved. 

\paragraph{Contribution.} In this paper, we take a first step towards addressing this challenge. We focus on the analysis of the following sequential process, inspired by the MultiQueue strategy: 
we are given $n$ queues, into which a large number of consecutively labeled elements are inserted initially, uniformly at random. 
To remove an element, we pick two queues at random, and remove the element of lower label on top of either queue. We are interested in the \emph{cost} of removals, defined as the \emph{rank} of the removed label among labels still present in any of the queues, i.e., the ``distance" from the optimal choice at each step, which would be the top element among all queues. 

Our main technical contribution is showing that this process provides surprisingly strong rank guarantees: 
for any time $t \geq 0$ in the execution of the process, the expected rank of an element removed at time $t$ is $O( n )$, while the expected worst-case rank removed at a step is $O( n \log n )$. These bounds are asymptotically tight, and hold for arbitrarily large $t$. Our analysis generalizes to a $(1 + \beta)$ extension of the process, where the algorithm deletes from the higher priority among two random queues with probability $0 < \beta < 1$, and from a single randomly chosen queue with probability $(1 - \beta)$. It also extends to show that the process is \emph{robust to bias} in the insertion distribution towards some bins by a constant factor $\gamma \in (0, 1)$. 
By contrast, we show that the strategy which always removes from \emph{a single} randomly chosen queue \emph{diverges}, in the sense that its average rank guarantee evolves as $\Omega\left( \sqrt{t  n \log n} \right)$, for time $t \geq n \log n$.  
The $(1 + \beta)$ strategy can be extended to an efficient implementation, which improves upon the state-of-the-art in terms of practical performance. 

\paragraph{Technical Overview.} A tempting first approach at analysis is to reduce to classic ``power of two choices" processes, e.g.~\cite{ABKU, PTW15}, in which balls are always inserted into the least loaded of two randomly chosen bins. 
A simple reduction between these two processes exists for the case where labels are inserted  into the queues in \emph{round-robin} fashion (please see Section~\ref{sec:reduction-rr} for details). 
However, this reduction breaks if elements are inserted uniformly at random, as is the case in real systems. 
Another important difference from the classic process is that elements are \emph{labelled}, and so the state of the system at a given step is highly correlated with previous steps. 
For instance, given a queue $Q$ whose top element has label $\ell$, if we wish to characterize the probability that the next label on top of $Q$ is a specific value $\ell' > \ell$, 
we need to know whether the history of elements examined by the algorithm (in other queues)  contains element $\ell'$. 
Such correlations make a standard step-by-step analysis extremely challenging. 

We circumvent these issues by relating the original labelled process to a continuous \emph{exponential process}, which reduces correlations by replacing integer labels with real-valued labels. 
In this process, on every ``insertion," we insert a new label whose value equals that of the latest inserted label in that bin (or $0$ for the first insertion), plus \emph{an exponential random variable} of mean $n$. 
Intuitively, we wish to simulate the integer label insertion process with the same mean $n$, while removing correlations but using  continuous label values to prevent label clashes and reduce correlations.  
Once this insertion process has ended, we replace each (real-valued) label with its \emph{rank} among all (real-valued) labels. 

We relate these two processes by the following key claim: \emph{the rank distributions for the original process and for the exponential process are identical}. 
More precisely, for any queue $i$ and rank $j$, the event that, after insertion, the label with rank $j$ is located in queue $i$ has the same probability in the original process and in the exponential process (and all such events are independent in both processes).

Given this fact, we can  couple the two processes as follows. We first couple the insertion steps such that each bin contains the same sequence of ranks in both processes. 
We then couple the removal steps so that they make exactly the same sequence of random choices. Under this coupling, the two processes will pay the same rank cost. 
Thus, it will suffice to characterize the rank cost of the exponential process under our removal scheme. 

This is achieved via two steps. 
The first is a potential argument which carefully characterizes the average value of labels on top of the queues, and the maximum deviation from the average by a queue,  at any time $t > 0$. 
This part of the proof is very technical, and generalizes an analytic approach developed for weighted balls-into-bins processes by Peres, Talwar and Wieder~\cite{PTW15}. 
Specifically, if we let $x_i(t)$ be the difference between the label on top of queue $i$ at time $t$ and the mean top label across all queues, 
then the \emph{total potential} of the system at time $t$ is defined as $\sum_{i = 1}^{n} \left[ \exp( \alpha x_i / n ) + \exp( - \alpha x_i / n ) \right]$. 

The key claim in the argument is showing that the expected value of this potential is bounded by $O( n )$ for any step $t$. This is done by a careful technical analysis of the expected change in potential at a step, 
which proves that the potential has supermartingale-like behavior: that is, it tends to decrease in expectation once it surpasses the $O(n)$ threshold. 
It is interesting to note that neither of the two exponential factors in the definition of the potential satisfies this bounded increase property in isolation, yet their sum can be bounded in this way.  

The $O(n)$ bound on the total potential provides a strong handle on the maximum deviation from the mean in the exponential process. 
The last step of our argument builds on this characterization to show that, since label values cannot stray too far either above or below the mean, the \emph{ranks} of these values among all values on top of queues cannot be too large. In particular, the average rank of elements on top of queues must be $O( n )$, and the maximum rank is $O( n \log n )$. The rank equivalence theorem implies our main claim. 

The analytic framework we sketched above is quite general. It can be extended to showing that this implementation strategy is robust to insertion bias, i.e. that it guarantees similar rank bounds even if the insertion process is biased (within some constant $\gamma \ll 1$) towards some of the bins. 
Further, it also applies to the $(1 + \beta)$ variant of the process, where each removal considers two randomly chosen queues with probability $\beta < 1$, and a single queue otherwise. 
Formally, the expected cost at a step is $O( n / \beta^2)$, and the expected maximum cost at a step is $O( n \log n / \beta)$. 

\paragraph{Applications.} 
The process considered above is sequential, and has FIFO semantics, since we are inserting labels in strictly increasing order, and assumes that all labels are inserted initially. We relax each of these assumptions, and discuss sufficient conditions for extending the analysis to concurrent and general priority insertions. 
To validate the practicality of our results, we implemented a variant of the $(1 + \beta)$ priority queue in C++. Empirical results, given in Section~\ref{appendix:tests}, show that it improves upon the state-of-the-art implementation of Rihani et al.~\cite{MQ} by $20\%$ in terms of throughput, while providing similar rank guarantees. Moreover, when used in Dijkstra's algorithm,  $(1 + \beta)$ priority queue reduces running time by $10\%$ over the original MultiQueue, and by $40\%$ over a deterministic implementation~\cite{klsm}. This suggests that the decrease in rank guarantees is well compensated by improved scalability.


\section{Related Work}

The first instance of a distributed data structure implementation using a similar approach to the one we consider is probably the parallel branch-and-bound framework by Karp and Zhang~\cite{KarZha93}. 
In their construction, each processor is assigned a queue, and elements are inserted randomly into one of the queues. Each processor removes from its own queue. 
This critically relies on the synchronicity of PRAM processors to achieve bounds on average rank and maximum rank difference. 
It is easy to see that processor delays can cause the rank difference to become unbounded. 
Quantitative relaxations of  concurrent data structures have been previously formalized by Henzinger et al.~\cite{Henzinger}, but only for \emph{deterministic} relaxations. 

\paragraph{The MultiQueue.} The strategy we consider is based on the  MultiQueue data structure of Rihani et al.~\cite{MQ}. 
They start from $n$ sequential priority queues, each of which is protected by a lock. 
To \emph{insert} an element, each processor repeatedly picks a queue at random, and tries to lock it. 
When successful, the processor inserts the element into the queue, and unlocks it. 
Otherwise, the process re-tries with a new randomly chosen lock. 

To \emph{delete} an element, the algorithm samples two queues uniformly, and reads the top element from both.
It then locks the queue whose top element had higher priority, removes that element, and returns it.
If the algorithm failed to obtain the lock on the chosen queue, then it restarts the whole process. 
(They also consider variations of this strategy, which are essentially equivalent from the point of view of the sequential version.)

In the paper, the authors show via a basic balls-into-bins argument that, initially, before any elements are removed, the expected rank cost is $O( n )$, and the max cost is $O( n \log n )$ with high probability. 
Our argument applies to a sequential version of the concurrent process described above. The argument is significantly more general, since it applies to any time in the execution of the process. 
Variants of this approach have also been considered in other contexts, for instance for relaxed queues~\cite{Haas} and priority schedulers~\cite{Nguyen13}. Designing efficient concurrent priority queues is a very active research area, see e.g.~\cite{Basin11, LotanShavit, klsm, SprayList} for recent examples, and~\cite{SprayList} for a survey. 

\paragraph{Balls-into-bins Processes.} 
In the classic two-choice balanced allocation process, at each step, one new ball is to be inserted into one of $n$ bins; the ball is inserted into the less loaded among two random choices~\cite{ABKU, Mitz01}. 
In this setting, it is known that the most loaded bin is at most $O( \log \log n )$ above the average.  
The literature studying extensions of this process is extremely vast; we direct the reader to~\cite{Richa01} for a survey. 
Considerable effort has been dedicated to understanding guarantees in the ``heavily-loaded" case, where the number of insertion steps is unbounded~\cite{Berenbrink00, PTW15}, 
and in the ``weighted" case, in which ball weights come from a probability distribution~\cite{TW07, Berenbrink08}. 
In a tour-de-force, Peres, Talwar, and Wieder~\cite{PTW15} gave a potential argument characterizing a general form of the heavily-loaded, weighted process on graphs. 
The second step in our proof, which characterizes the deviation of weights from the mean in the exponential process, builds on their approach. 
Specifically, we use a similar potential function, and the parts of the analysis are similar. 
However, we generalize their approach to the case of biased insert distributions, and our argument diverges in several technical points: in particular, the potential analysis under unbalanced conditions is different. 
The first and third steps of our argument are completely new, as we consider a more complex \emph{labelled} allocation process.

\section{Process Definition} 
\label{sec:def}

%


We are given with $n$ priority queues, labeled $1, \ldots, n$. 
To insert,  we choose queue $i$ with probability $\pi_i \in (0, 1)$, and insert into that queue.
We assume that $\sum_{i = 1}^n \pi_i = 1$, and that the bias is bounded, i.e. there exists a constant $\gamma \in (0, 1)$ such that, for any queue $i$, $1 - \gamma \leq \frac{1}{n\pi_i} \leq 1 + \gamma$. 
To remove, with probability $0 < \beta \leq 1$, we pick two queues u.a.r., and remove the element of minimum label among the two choices, and with probability $(1 - \beta)$ remove from a random queue. 
We call such a process a $(1 + \beta)$-sequential process.
When omit $\beta$ when unimportant or clear from context.

At any point in the execution, we define the \emph{rank} of any element to be the number of elements currently in the system which have lower label than it (including itself), so the minimal rank is $1$. 
Our goal is to show that the ranks of the elements returned by the random process are small, throughout the entire execution.
We will restrict our attention to executions which never inspect empty queues, and no priority inversions on inserts are visible to the removal process.
Formally:
\begin{definition}
An execution of the sequential process is \emph{prefixed} if except with negligible probability (i.e. probability less than $n^{-\omega (1)}$), the sequential process never inspects an empty queue in a remove, and no remove sees an insert with lower priority than one that is inserted later in the same queue. 
\end{definition}
There are many natural types of executions which are prefixed. For instance, a common strategy in practice is to insert a large enough ``buffer" of elements initially, so as to minimize the likelihood of examining empty queues. 
More formally, since in each iteration we touch a queue with probability at most $\frac{2}{n} + O(\frac{1}{n^2})$ and only remove one element at a time, by Bernstein's inequality, any execution of length $T$ whose first $T / 2 + \sqrt{n \cdot \omega (\log n)}$ operations are inserts, and which contains no inserts after this time, or which only inserts larger labels after this point, will be prefixed.

For the rest of the paper, we will tacitly assume that no remove inspects any empty queue in the sequential process, and that no priority inversions are visible to the removal process, and condition on the event that this does not happen.
By Markov's inequality, as long as we only consider prefixed executions, then for any sequence of operations of length $\poly (n)$, this will change the max rank or average rank of the operations by at most a subconstant factor.

\section{Analysis of the Sequential Process}
The goal of this section is to prove the following theorem:
\begin{theorem}[see Corollary \ref{cor:max} and Corollary \ref{cor:avg}]
\label{thm:main}
Fix a bias bound $\gamma$.
In any prefixed $(1 + \beta)$ sequential process, for $\beta = \Omega(\gamma)$, and for any time $t > 0$, the max rank of any element on top of a queue at time $t$ is at most $O(\frac{1}{\beta} (n \log n + n \log \frac{1}{\beta}))$, and the average rank at time $t$ is at most $O(n / \beta^2)$.
\end{theorem}


\label{sec:analysis-general}
\paragraph{The Exponential Process.}
For the purposes of analysis, it is useful to consider an alternative process, which we will call the \emph{exponential process}, in which each bin independently generates real-valued labels by starting from $0$, and adding an exponentially-distributed random variable with parameter $\lambda_i = 1/\pi_i$ to its previous label. 
More precisely, if $w_i(t)$ is the label present on top of bin $i$ at time $t$, then the value of the next label is $w_i(t + 1) = w_i(t) + \exp(1 / \pi_i)$.
Once we have enqueued all the items we wanted to enqueue, we proceed to remove items by the two-choice rule: we always pick the element of minimum rank (label value) among two random choices. 
At any step, and for any $v \in \mathbb{R}$, we define $\rank (v)$ to be the number of elements currently in the system with label at most $v$.
At each step, we pay a cost equal to the rank of the element we just removed.

\paragraph{Proof Strategy.} 
The argument proceeds in three steps.  
First, we show that, perhaps surprisingly, after all insertion steps have been performed, the rank distribution of the exponential process is \emph{the same} as the rank distribution of the original process (Section~\ref{sec:equivalence}). 
With this in place, we will perform the following coupling between the two processes:

Given an arbitrary number $M$ of balls to be inserted into the $n$ bins, we first generate $M$ total weights across the $n$ queues in the exponential process, so that each bin has exactly the same number of elements in both processes. 
We then replace each (real-valued) generated weight with its \emph{rank} among all weights in the system. 
As we will see, for any rank $j$ and queue $i$, the probability that the ball with rank $j$ is in bin $i$ is exactly the same as in the simple process, where queue $i$ has probability of insertion $\pi_i$. 
More precisely, if we fix an increasing  sequence of ranks $r_1, r_2, \ldots$, the probability that bin $i$ contains this exact sequence of ranks is the same across the two processes. 
Hence, we can couple the two processes to generate the same sequence of ranks in each bin. 

For the removal step, we couple the two processes such that, in every step, they generate exactly the same $\beta$ values and random choice indices. 
Since the ranks are the same, and the choices are the same, the two processes will remove from exactly the same bins at each step, and will pay the same rank cost in every step. 

Hence, it will suffice to bound the expected rank cost paid at a step in the exponential process. Since this is difficult to do directly, we will first aim to characterize the \emph{concentration} of the difference between the weight (value) on top of each bin and the average weight on top of bins, via a potential argument (Section~\ref{sec:potential}). 
This will show that relatively few values can stray far from the mean. In turn, this will imply a bound on the average rank removed (Section~\ref{sec:avgrank}). 

In the following, we use the terms ``queue" and ``bin" interchangeably. 

\subsection{Equivalence between Rank Distributions} 
\label{sec:equivalence} 

Let $\pi_i$ be the probability that a ball is inserted into bin $i$.  This section will be dedicated to showing that the  distribution of ranks in the exponential process where bin $i$ gets weights exponentially distributed with mean $1 / \pi_i$ is identical to the distribution of labels in the original process, where the $i$th bin is chosen for insertion with probability $\pi_i$.

\begin{theorem}
\label{lem:distributions}
	Let $I_{j \gets i}$ be the event that the label with rank $i$ is located in bin $j$, in either process. 
	Let $\Pr_e [ I_{j \gets i} ]$ be its probability in the exponential process, and $\Pr_o [ I_{j \gets i} ]$ be its probability in the original process. Then for both processes $I_{j\gets i}$ is independent from $I_{j'\gets i'}$ for all $i\neq i'$ and
$$
		\Pr_e [ I_{j \gets i}] = \Pr_o [ I_{j \gets i} ] = \pi_j.
$$
\end{theorem}
\begin{proof}
That $\Pr_o  [ I_{j \gets i} ] = \pi_j$ and that $I_{j\gets i},I_{j'\gets i'}$ are independent for $i\neq i'$ in the original process both follow immediately by the construction of the sequential process.

We now require some notation.
Let $\ell(i)$ be the label with rank $i$ in the exponential process, and let $b(i)$ be the bin containing $\ell(i)$. 
We will employ the following \emph{memoryless} property of exponential random variables:
\begin{fact}
Let $X \sim \Exp (r)$.
Then, for all $s, t \geq 0$, we have $\Pr [X > s + t] = \Pr [X > s] \Pr [X > t]$.
\end{fact}

Fix an arbitrary $i \geq 1$, and let $L = \ell({i - 1})$ be the label of the $(i - 1)$th element. 
For each bin $i$, we isolate the two consecutive weights between which $L$ can be placed. More precisely, for each bin $1\leq j \leq n$, 
denote the smallest label \emph{larger} than $L$ in bin $j$ by $\ell_{j,>L}$ and let the largest label in $B_j$ \emph{smaller} than $L$ be $\ell_{j,<L}$. 
By construction of the exponential process, $\ell_{j,>L} = \ell_{j,<L} + X_j$ where $X_j$ is an exponentially distributed random variable with mean $1/\pi_j$. Furthermore, by assumption we have that $\ell_{j,>L} > L$, so $X_j > L - \ell_{j,<L}$. 
 Then by memoryless-ness, 
\begin{eqnarray*}
\Pr[\ell_{j,>L} > \ell_{j,<L} + s | \ell_{j,>L} > L] & = & \Pr[X_j > s | X_j > (L-\ell_{j,<L})]\\
& = &\Pr[X_j > (L-\ell_{j,<L}) + (\ell_{j,>L} - L) | X_j > (L-\ell_{j,<L})]\\
& = &\Pr[X_j > \ell_{j,>L} - L]
\end{eqnarray*}
Let $\Delta_j = \ell_{j,>L} - L$, observing by the above that $\Delta_j$ is distributed exponentially with mean $\pi_i$. Furthermore, since the $X_j$ are independent, the $\Delta_j$ are independent as well. Consider the event that $\ell_{j,>L}$ is the smallest label larger than $L$ in the system; that is $\ell(i) = \ell_{j,>L}$ and $b(i) = j$.
This occurs if $\Delta_j < \Delta_j'$ for all $j'\neq j$, i.e. if $\Delta_j$ is the smallest such value among all bins. 
In turn, this occurs with probability

\begin{eqnarray*}
	\int_{0}^{\infty} \Pr[ \Delta_j = t ] \Pi_{k \neq j} \Pr[ \Delta_k > t ] dt & = &	\int_{0}^{\infty} \pi_j \exp\left( - t \pi_j \right)  \Pi_{k \neq j} \exp\left( - t \pi_k \right) dt \\
	& = &\pi_j \int_{0}^{\infty} \exp\left( - t \sum_{k = 1}^n \pi_k \right) dt = \pi_j.
\end{eqnarray*}

\end{proof}

\subsection{Analysis of the Exponential Process}
\label{sec:potential}


In the previous section, we have shown that the rank distributions are the same at the end of the insert process. The coupling described in Section~\ref{sec:analysis-general}, by which we inspect the same queues in each removal steps,  
implies that two processes will produce the same cost in terms of average rank removed. 
Hence, we focus on bounding the \emph{rank} of the labels removed from the exponential process at each step $t$. 
Our strategy will be to first bound the deviation of the \emph{weight} on top of each queue from the average, at every step $t$, and then re-interpret the deviation in terms of rank cost.  

\paragraph{Notation.}
From this point on, we will assume for simplicity that, at the beginning of each step, queues are always ranked in \emph{increasing order of their top label}. If $p_i$ is the probability that we pick the $i$th ranked bin for a removal, and $\beta$ is the two-choice probability, then it is easy to see that the $(1 + \beta)$ process guarantees that

$$p_i =  (1 - \beta)\frac{1}{n} + \beta \left[\frac{2}{n} \left(1 - \frac{i - 1}{n} \right) - \frac{1}{n^2}\right].$$ 

\noindent Further, notice that, for any $1 \leq m \leq n$, we have, ignoring the negligible $O( 1 / n^2)$ factor, that $$\sum_{i = 1}^{m} p_i \simeq \frac{m}{n} \left(  1 + \beta - \frac{m}{n} \beta \right).$$

\noindent For any bin $j$ and time $t$, let $w_j(t)$ be the label on top of bin $j$ at time $t$, and let $x_j(t) = w_j(t) / n$ be the normalized label. 
Let $\mu(t) = \sum_{j = 1}^{n} x_j(t) / n$ be the average normalized label at time $t$ over the bins. 
Let $y_i(t) = x_i(t) - \mu (t)$, and let $\alpha < 1$ be a parameter we will fix later. Define  $$\Phi(t) = \sum_{j = 1}^n \exp\left( \alpha y_i(t) \right), \,\text{and}\, \Psi(t) = \sum_{j = 1}^n \exp\left( - \alpha y_i(t) \right).$$  

\noindent Finally, define the potential function 
$$ \Gamma(t) = \Phi(t) + \Psi(t).$$ 

\paragraph{Parameters and Constants.} Define $\epsilon = \frac{\beta}{16}$. Recall that the parameter $\gamma >0 $ is such that, 
for every $1 \leq i \leq n$, $\frac{1}{n\pi_i} \in [1 - \gamma, 1 + \gamma]$. 
Next, let $c \geq 2$ be a constant, and $\delta$ be a parameter such that 
\begin{eqnarray}
	1 + \delta = \frac{ 1 + \gamma + c \alpha \left( 1 + \gamma \right)^2} { 1 - \gamma - c \alpha \left( 1 + \gamma \right)^2  }. \label{eq:param-delta}
\end{eqnarray}
In the following, we assume that the parameters $\alpha, \beta,$ and $\gamma$ satisfy the inequality 
\begin{equation} \frac{\beta}{16} = \epsilon \geq \delta. \label{eq:param-beta}
\end{equation}
\noindent
We assume that the insertion bias $\gamma \leq 1 / 2$ is small, and hence this is satisfied by setting $\beta = \Omega(\gamma)$ and $\alpha = \Theta(\beta)$.

\noindent The rest of this section is dedicated to the proof of the following claim:
\begin{theorem}
\label{thm:beta}
	Let $\vec{p} = (p_1, p_2, \ldots, p_n)$ be the vector of probabilities, sorted in increasing order. 
	Let $\epsilon = \beta / 16$, and $c \geq 2$ be a small constant. 
	Let $\alpha$, $\beta$, $\gamma$, $\delta$ be parameters as given above, such that $\epsilon \geq \delta$. 
	Then there exists a constant $C(\epsilon) = \poly ( \frac{1}{\epsilon} )$ such that, for any time $t\geq 0$, we have $\E[ \Gamma(t) ] \leq  C(\epsilon) n$.
\end{theorem}

\paragraph{Potential Argument Overview.} 
The argument proceeds as follows. We will begin by bounding the change of each potential function $\Psi$ and $\Phi$, at each queue in a step (Lemmas~\ref{lem:phibound} and~\ref{lem:psibound}).
We then use these bounds to show that, if not too many queues have weights above or below the mean ($y_{n / 4} \leq 0$ and $y_{3n/4} \geq 0$), then $\Phi$ and $\Psi$ respectively decrease in expectation (Lemmas~\ref{lem:phigood} and~\ref{lem:psigood}). 
Unfortunately, this does not necessarily hold in general configurations. 
However, we are able to show the following claim: if the configuration is unbalanced (e.g. $y_{n / 4} > 0$ ) and $\Phi$ does not decrease in expectation at a step, then either the symmetric potential $\Psi$ is large, and will decrease on average, or the global potential function $\Gamma$ must be in $O( n )$ (Lemma~\ref{lem:phibad}). We will prove a similar claim for $\Psi$. Putting everything together, we get that the global potential function $\Gamma$ always decreases in expectation once it exceeds the $O(n)$ threshold, which implies Theorem~\ref{thm:beta}.

\paragraph{Potential Change at Each Step.}
We begin by analyzing the expected change in potential for each queue from step to step. 
We first look at the change in the weight vector $\vec{y} = (y_1, y_2, \ldots, y_n)$. 
(Recall that we always re-order the queues in increasing order of weight at the beginning of each step.)
Below, let $\kappa_j$ be the cost increase if the bin of rank $j$ is chosen, which is an exponential random variable with mean $1 / \pi_j$. 
We have that, for every rank $i$, 

$$
y_i(t + 1) =  \begin{cases} 
y_i(t) + \kappa_i \left(1 - 1 / n\right)   , \mbox{ with probability } p_i \mbox{; i.e., the queue of rank $i$ is picked, and }\\
y_i(t) - \frac{\kappa_j }{n}, \mbox{if some other bin $j \neq i$ is chosen}. 
\end{cases}
$$

\paragraph{The Change in $\Phi$ at a Step.}
We prove the following bound on the expected  increase in $\Phi$ at a step. 

\begin{lemma}
	\label{lem:phibound}
	For any bin rank $i$, 
	$$\E \left[ \Delta \Phi \,|\, y(t) \right] \leq \sum_{i = 1}^n \hat{\alpha}\left(  (1 + \delta) p_i - \frac{1}{n}  \right) \exp\left( \alpha y_i(t) \right).$$
\end{lemma}

\begin{proof}
Let $\Phi_i (t) = \exp\left( \alpha y_i(t) \right)$. We have two cases. If the bin is chosen for removal, then the change is:
\begin{eqnarray*}
	\Delta \Phi_i  := \Phi_i(t+ 1) - \Phi_i(t) = & \exp\left( \frac{\alpha}{n} ( y_i(t) + \kappa_i \left(1 - \frac{1}{n} \right) ) \right) - \exp\left(  \frac{\alpha}{n} y_i(t)  \right) \\ 
	= & \exp\left( \frac{\alpha}{n} y_i(t) \right) \left( \exp \left( \frac{\alpha}{n} \kappa_i \left(1 - \frac{1}{n} \right) \right) - 1 \right).
\end{eqnarray*}

\noindent 

\noindent Taking expectations with respect to the random choices made on insertion, that is, the value of $\kappa_j$, we have 
\begin{eqnarray*}
	\E_i \left[  \exp \left( \kappa_i \frac{\alpha}{n}  \left(1 - \frac{1}{n} \right) \right)  \right] \stackrel{(a)}{=} & \frac{\pi_i }{ \pi_i - \left( \frac{\alpha}{n}  ( 1 - \frac{1}{n} ) \right)} = \frac{1}{ 1 - \frac{\alpha}{\pi_i n} \left(1 - \frac{1}{n} \right) } \stackrel{(b)}{\leq} 1 + \frac{\alpha}{\pi_i n} \left(1 - \frac{1}{n} \right) + c \left( \frac{\alpha}{\pi_i n} \left(1 - \frac{1}{n} \right) \right)^2,
\end{eqnarray*}
\noindent for some constant $c > 1$. Step (a) follows from the observation that the expectation we wish to compute is the moment-generating function of the exponential distribution at $\frac{\alpha}{n} ( 1 - \frac{1}{n} )$, while (b) follows from the Taylor expansion.  (We slightly abused notation in step (a) by denoting with $\pi_i$ the insert probability for the $i$th ranked queue, according to the ranking in this step.)

\noindent The second step if if some other bin $j \neq i$ is chosen for removal, then the change is:
\begin{eqnarray*}
	\Delta \Phi_i = & \exp\left( \frac{\alpha}{n} ( y_i(t) - \kappa_j \frac{1}{n}  ) \right) - \exp\left(  \frac{\alpha}{n} y_i(t)  \right) 
	= \exp\left( \frac{\alpha}{n} y_i(t) \right) \left( \exp \left( - \frac{\alpha \kappa_j}{n^2}  \right)  - 1 \right).
\end{eqnarray*}

\noindent Again taking expectations with respect to the random choices made on insertion, i.e. the value of $\kappa_j$, we have 
\begin{eqnarray*}
	\E_i \left[  \exp \left( - \kappa_j \frac{\alpha}{n^2}  \right)  \right] = & \frac{ \pi_j }{ \pi_j + \frac{\alpha}{n^2}   } = 
	\frac{ 1 }{ 1 +   \frac{\alpha}{\pi_j n^2} } \leq 1 -  \frac{\alpha}{\pi_j n^2} + \left( \frac{\alpha}{\pi_j n^2} \right)^2 - \ldots \leq 1 -  \frac{\alpha}{\pi_j n^2} + \left( \frac{\alpha}{\pi_j n^2} \right)^2.
\end{eqnarray*}

\noindent Therefore, we have that 
\begin{eqnarray*}
	\E \left[ \Delta \Phi_i \right] / \Phi_i (t) & \leq & \left( 1 + \frac{\alpha}{\pi_i n} \left(1 - \frac{1}{n} \right) + c \left( \frac{\alpha}{\pi_i n} \left(1 - \frac{1}{n} \right) \right)^2 \right) p_i + 
	\sum_{j \neq i} \left(1 -  \frac{\alpha}{\pi_j n^2} + \left( \frac{\alpha}{\pi_j n^2} \right)^2\right) p_j - 1 \\
	& \leq & \left( \frac{\alpha}{\pi_i n} + c \left( \frac{\alpha}{\pi_i n} \right)^2 \right) p_i - 
	\sum_{j = 1}^n \left( \frac{\alpha}{\pi_j n^2} - c \left( \frac{\alpha}{\pi_j n^2} \right)^2\right) p_j  \\ 
	& \leq & \left( 1 + \gamma + c \alpha \left( 1 + \gamma \right)^2 \right) \alpha p_i - 
	\frac{\alpha}{n} \sum_{j = 1}^n \left( 1 - \gamma - c \alpha \left( 1 + \gamma \right)^2\right) p_j.
%
%
\end{eqnarray*}

\noindent If we denote for convenience 
$$  \hat{\alpha} = \alpha \left( 1 - \gamma - c \alpha \left( 1 + \gamma \right)^2  \right), \text{ and recall that } \delta := \frac{ 1 + \gamma + c \alpha \left( 1 + \gamma \right)^2} { 1 - \gamma - c \alpha \left( 1 + \gamma \right)^2  } - 1,$$

\noindent then we can rewrite this as 
\begin{eqnarray*}
	\E \left[ \Delta \Phi \,| \, y(t) \right] \leq  \sum_{i = 1}^n \hat{\alpha} \left( p_i (1 + \delta) - 
	\frac{1}{n} \right) \Phi_i(t).
%
%
\end{eqnarray*}
\end{proof}

%

\paragraph{The Change in $\Psi$.}
Using a symmetric argument, we can prove the following about the expected change in $\Psi$. 

\begin{replemma}{lem:psibound}
	$$\E \left[ \Delta \Psi \,|\, y(t) \right] \leq \sum_{i = 1}^n \hat{\alpha}\left(  (1 + \delta) \frac{1}{n} - p_i  \right) \exp\left( - \alpha y_i(t) \right).$$
\end{replemma}

%

\paragraph{Bounds under Balanced Conditions.}
Let us briefly stop to examine the bounds in the above Lemmas. The terms $\left( p_i (1 + \delta) - \frac{1}{n} \right)$ are \emph{decreasing} in $i$, and in fact become \emph{negative} as $i$ increases. (The exact index where this occurs is controlled by $\beta$ and $\delta$.) The  $\exp\left( \alpha y_i(t) \right)$ terms are \emph{increasing} in $i$. 
Bins whose weight is \emph{below} the mean (i.e., $y_i (t) \leq 0$) have a negligible effect on $\Phi$, since each of their contributions is at most $1$. At the same time, notice that the contribution of bins of large index $i$ will be negative. 
Hence, we can show that, if at least $n / 4$ bins have weights below average, then the value of $\Phi$ tends to \emph{decrease} on average. 
\begin{replemma}{lem:phigood}
If $y_{n / 4} \leq 0$, then we have that 
\begin{eqnarray*}
\E \left[  \Phi(t + 1) \,|\, y(t) \right] 
\leq \left( 1 -   \frac{\hat{\alpha} \epsilon}{3n} \right) \Phi(t) + 1.
\end{eqnarray*}
\end{replemma}

%

\noindent A similar claim holds for $\Psi$, under the condition that there are at least $n / 4$ bins with weight \emph{larger} than average.

\begin{replemma}{lem:psigood}
If $y_{3n / 4} \geq 0$, then we have that 
\begin{eqnarray*}
\E \left[ \Psi (t + 1) \,|\, y(t) \right] \leq &  \left(   1 -  \frac{\hat{\alpha} \epsilon}{3n}  \right) \Psi + 1.
\end{eqnarray*}
\end{replemma}

%
%
%

\paragraph{Bounds under Unbalanced Conditions.}
We now analyze unbalanced configurations, where there are either many bins whose weights are above average (e.g., $y_{n / 4} > 0$), or below average ($y_{3 n / 4} < 0$).  
The rationale we used to bound each potential function independently no longer applies. 
In particular, as shown in~\cite{PTW15}, we can have unbalanced settings where for example $\Phi$ does not decrease in expectation. 
However, we can show that,  one of two things must hold: 
either the other potential $\Psi$ is \emph{larger} and \emph{does decrease} in expectation, or the global potential $\Gamma$ is in $O(n)$. 

\begin{replemma}{lem:phibad}
Given $\epsilon$ as above, assume that $y_{n / 4}(t) > 0$, and  $\mathbb{E}[ \Delta \Phi ] \geq - \frac{\epsilon \hat{ \alpha}}{ 3n } \Phi(t)$. 
 Then either $\Phi < \frac{\epsilon}{4} \Psi$ or $\Gamma = O( n )$.   
\end{replemma}
\begin{proof}

Fix $\lambda  = 2 / 3 - 1 / 54$  for the rest of the proof. 
We can split the inequality in Lemma~\ref{lem:phibound} as follows:
\begin{eqnarray}
\E \left[ \Delta \Phi \,|\, y(t) \right] \leq \sum_{i = 1}^{\lambda n} \hat{\alpha}\left(  (1 + \delta) p_i - \frac{1}{n}  \right) \exp\left( \alpha y_i(t) \right) + 
\sum_{i = \lambda n + 1}^{n} \hat{\alpha}\left(  (1 + \delta) p_i - \frac{1}{n}  \right) \exp\left( \alpha y_i(t) \right). 
\label{ineq:phisplit}
\end{eqnarray}

\noindent We bound each term separately. Since the probability terms are non-increasing and the exponential terms are non-decreasing, the first term is maximized when all $p_i$ terms are equal. Since these probabilities are at most $1$, we have
\begin{eqnarray}
\label{ineq:philess}
\sum_{i = 1}^{\lambda n} \hat{\alpha}\left(  (1 + \delta) p_i - \frac{1}{n}  \right) \exp\left( \alpha y_i(t) \right) \leq \frac{\hat{\alpha}}{n} \left( (1 +\delta) \frac{1}{\lambda} - 1 \right) \Phi_{\leq \lambda n}.
\end{eqnarray} 

\noindent The second term is maximized by noticing that the $p_i$ factors are non-increasing, and are thus dominated by their value at $\lambda n$. Noticing that we carefully picked $\lambda$ such that 
$$ p_{\lambda n} \leq \frac{1}{n} - \frac{4 \epsilon}{n}, $$
we obtain, using the assumed inequality $\delta \leq \epsilon$, that 
\begin{eqnarray}
\label{ineq:phimore}
	\sum_{i = \lambda n + 1}^{n} \hat{\alpha}\left(  (1 + \delta) p_i - \frac{1}{n}  \right) \exp\left( \alpha y_i(t) \right)  \leq \hat{\alpha} \left( \delta - 4 \epsilon  \right) \frac{\Phi_{>\lambda n}}{n} \leq - \frac{3 \epsilon \hat{\alpha}}{n} \Phi_{>\lambda n}. 
\end{eqnarray}

\noindent By the case assumption, we know that $
\E \left[ \Delta \Phi \,|\, y(t) \right] \geq - \frac{\hat{\alpha} \epsilon}{3n} \Phi(t).$
 Combining the bounds (\ref{ineq:phisplit}), (\ref{ineq:philess}), and (\ref{ineq:phimore}), this yields:
\begin{eqnarray*}
\frac{\hat{\alpha}}{n} \left( (1 +\delta) \frac{1}{\lambda} - 1 \right) \Phi_{\leq \lambda n}  - \frac{3 \epsilon \hat{\alpha}}{n} \Phi_{>\lambda n} 
 \geq - \frac{\hat{\alpha}\epsilon}{3n} \Phi(t). 
\end{eqnarray*}

\noindent Substituting $\Phi_{>\lambda n} =  \Phi - \Phi_{\leq \lambda n}$ yields:
%
%
\begin{eqnarray*}
 \left(  3 \epsilon - \epsilon / 3 \right) \Phi \leq \left(  (1 + \delta) \frac{1}{\lambda} - 1 + 3 \epsilon  \right) \Phi_{\leq \lambda n} . 
\end{eqnarray*}

\noindent For simplicity, we fix 
$
 C(\epsilon) =  \frac{  (1 + \delta) \frac{1}{\lambda} - 1 + 3 \epsilon }{ 3 \epsilon - \epsilon / 3 } = O\left( \frac{1}{\epsilon} \right), 
$
to obtain 
\begin{eqnarray}
\label{eq:phirel}
\Phi \leq C(\epsilon)  \Phi_{\leq \lambda n}. 
\end{eqnarray}

\noindent Let $B = \sum_{y_i > 0} y_i$. Since we are normalizing by the mean, it also holds that $B = \sum_{y_i < 0} (- y_i )$. 
Notice that 
\begin{eqnarray}
\label{eq:brel}
\Phi_{ \leq \lambda n} \stackrel{y_i \text{incr.}}{\leq} \lambda n \exp{\left(  {\alpha y_{\lambda n} } \right)} 
 \stackrel{y_i \text{incr.}}{\leq} \lambda n \exp{\left(  \frac{\alpha B }{ (1 - \lambda) n} \right)}.
\end{eqnarray}

\noindent We put inequalities~(\ref{eq:phirel}) and~(\ref{eq:brel}) together and get 
\begin{eqnarray}
\label{eq:phitwo}
	\Phi(t) \leq & \lambda n C(\epsilon)  \exp{\left(  \frac{\alpha B }{ (1 - \lambda) n} \right)} ,
\end{eqnarray}

\noindent Let us now lower bound the value of $\Psi$ under these conditions. Since $y_{n / 4} > 0,$ all the costs below average must be in the first quarter of $y$. 
We can apply Jensen's inequality to the first $n / 4$ terms of $\Psi$ to get that 
\begin{eqnarray*}
	\Psi \geq \sum_{i = 1}^{n / 4} \exp \left( - \alpha {y_i} \right) \geq 
	\frac{n}{4} \exp \left(  - {\alpha} \frac{\sum_{i = 1}^{n / 4} y_i }{  n / 4  }      \right).
\end{eqnarray*}

\noindent We now split the sum $\sum_{i = 1}^{n / 4} y_i$ into its positive part and its negative part. We know that the negative part is summing up to exactly $- B$, as it contains all the negative $y_i$'s and the total sum is $0$. The positive part can be of size at most $B / 4$,  since it is maximized when there are exactly $n - 1$ positive costs and they are all equal. 
Hence the sum of the first $n / 4$ elements is at least $ - 3B / 4$, which implies that the following bound holds: 
\begin{eqnarray}
\label{eq:psilb}
	\Psi \geq 
	\frac{n}{4} \exp{ \left(  - {\alpha} \frac{- 3B / 4 }{  n / 4  } \right)} \geq 
	\frac{n}{4} \exp \left(  {\alpha} \frac{ 3B }{ n } \right). 
\end{eqnarray}

\noindent If  $\Phi < \frac{\epsilon}{4} \Psi$, then there is nothing to prove. 
Otherwise, if $\Phi \geq  \frac{\epsilon}{4} \Psi$, we get from~(\ref{eq:psilb}) and~(\ref{eq:phitwo}) that 
\begin{eqnarray*}
	\frac{\epsilon}{4} \frac{n}{4} \exp \left(  {\alpha} \frac{ 3B }{ n } \right) \leq \frac{\epsilon}{4} \Psi  \leq  \Phi(t) \leq & \lambda n C(\epsilon)  \exp{\left(  \frac{\alpha B }{ (1 - \lambda) n} \right)} ,
\end{eqnarray*}
\noindent Therefore, we get that: 
\begin{eqnarray*}
	 \exp \left(  {\alpha} \frac{ B }{ n } \left(  3 - \frac{1}{1 - \lambda}  \right) \right) \leq \frac{4\lambda}{\epsilon} C(\epsilon) = O\left( \frac{1}{\epsilon^2} \right). 
\end{eqnarray*}

\noindent Using the mundane fact that $  3 - \frac{1}{1 - \lambda} = 3 / 19$, we get that 
%
\begin{eqnarray}
\label{Bbound}
	 \exp \left(   \frac{ {\alpha}B }{ n } \right) \leq O\left( \frac{1}{\epsilon^{14}}\right). 
\end{eqnarray}

\noindent To conclude, notice that (\ref{Bbound}) implies we can upper bound $\Gamma$  in this case as:
 \begin{eqnarray*}
 \Gamma  = \Phi + \Psi \leq & \frac{4 + \epsilon}{\epsilon} \Phi \leq \frac{4 + \epsilon}{\epsilon} \lambda n C(\epsilon)  \exp{\left(  \frac{\alpha B }{ (1 - \lambda) n} \right)}
 \leq  O\left( \frac{1}{\epsilon^{14/(1 - \lambda)}}\right) \frac{4 + \epsilon}{\epsilon}  C(\epsilon) \lambda n = O( \text{poly}\left( \frac{1}{\epsilon} \right ) n ). 
 \end{eqnarray*}

\end{proof}

\noindent We can prove a symmetric claim for $\Psi$ by a slightly different argument. 
\begin{replemma}{lem:psibad}
Given $\epsilon$ as above, assume that $y_{\frac{3n}{4}} < 0$, and that $\mathbb{E}[ \Delta \Psi ] \geq -\frac{ \hat{\alpha} \epsilon}{ 3 n } \Psi$. 
	Then either $\Psi < \frac{\epsilon}{4} \Phi$ or $\Gamma = O (n )$. 
\end{replemma}

\paragraph{Endgame.} We now finally have the required machinery to prove that $\Gamma$ satisfies a supermartingale property:

\begin{lemma}
\label{lem:gammabound}
There exists a constant $\epsilon$ such that $$\mathbb{E}[ \Gamma(t + 1) | y(t) ] \leq \left( 1 - \frac{ \hat{\alpha} \epsilon }{4n}\right) \Gamma(t) + C, \textnormal{ where $C$ is a constant in $O\left( \text{poly} \left( \frac{1}{\epsilon} \right)\right)$.}$$ 
\end{lemma}
\begin{proof}

\paragraph{Case 1:} If $y_{n / 4} \leq 0$ and $y_{3n/4} \geq 0$, then the property follows by putting together Lemmas~\ref{lem:phigood} and~\ref{lem:psigood}. 

\paragraph{Case 2:} If $y_{3n/4} \geq y_{n / 4}  \geq 0$. This means that the weight vector is unbalanced, in particular that there are few bins of low cost, and many bins of high cost. 
However, we can show that the expected decrease in $\Psi$ can compensate this decrease, and the inequality still holds. Notice that Lemmas~\ref{lem:phibound} and~\ref{lem:psibound} imply that $\mathbb{E}[ \Delta \Phi ] \leq  \frac{\hat{\alpha} \delta}{n}  \Phi(t)$ and $\mathbb{E}[ \Delta \Psi ] \leq  \frac{\hat{\alpha} \delta}{n}  \Psi(t)$, respectively. 

If $\mathbb{E}[ \Delta \Phi ] \leq - \frac{\epsilon \hat{ \alpha }}{ 4 n } \Phi(t)$, then the claim simply follows by putting this bound together with Lemma~\ref{lem:psigood}. 
Otherwise, if $\mathbb{E}[ \Delta \Phi ] \geq - \frac{\epsilon \hat{\alpha} }{ 4 n } \Phi(t)$, we obtain from Lemma~\ref{lem:phibad} that either $\Phi < \frac{\epsilon}{4} \Psi$ or $\Gamma = O(n)$. 
In the first sub-case, we have that 
\begin{eqnarray*} 
\mathbb{E} [ \Delta \Gamma ] = \mathbb{E} [ \Delta \Phi ] + \mathbb{E} [ \Delta \Psi ] \leq & 
 \frac{\hat{\alpha} \delta}{n}  \Phi(t) + 1 - \frac{\hat{\alpha}\epsilon}{3n} \Psi(t) \leq 
- \frac{\hat{\alpha} \epsilon }{3n}  \Psi(t) + 1 + \frac{\epsilon \delta \hat{\alpha}}{4n} \Psi(t)  \\ \leq & - \frac{\hat{\alpha} \epsilon}{3n}  \left( 1 - \frac{3\delta}{4}  \right) \Psi(t) + 1 \leq - \frac{\hat{\alpha}}{3n} \left( 1 + \frac{\epsilon}{4}  \right)^{-1} \left( 1 - \frac{3\delta}{4}  \right)  \Gamma(t) \leq - \frac{\hat{\alpha}}{4 n} \Gamma(t)
, \text{as claimed.} \end{eqnarray*}

\noindent In the second sub-case, we know that $\Gamma \leq C n$, for some constant $C$. 
Hence,  we can get that 
$$\mathbb{E}[ \Delta \Gamma ]  = \mathbb{E} [ \Delta \Phi ] + \mathbb{E} [ \Delta \Psi ] \stackrel{(a)}{\leq} \frac{\hat{\alpha} \delta}{n} \Gamma \stackrel{\Gamma = O(n)}{\leq} 2C \hat{\alpha},$$ 

\noindent where in step (a) we used the upper bounds in Lemmas~\ref{lem:phibound} and~\ref{lem:psibound}, respectively.
On the other hand, 
$$ C - \frac{\hat{\alpha}\epsilon}{4 n } \Gamma(t)  \geq  C - \frac{\hat{\alpha}\epsilon}{4 n } C n = C \left( 1 - \frac{\hat{\alpha}\epsilon}{4 }\right)  \geq  3C \hat{\alpha} \geq \mathbb{E}[ \Delta \Gamma ].$$

\paragraph{Case 3:} If $y_{n / 4} \leq y_{3n / 4} < 0$. 
This case is symmetric to the one above. 
\end{proof}

\noindent The intuition behind the above bound is that $\Gamma$ will always tend to decrease once it surpasses the $\Theta(n)$ threshold. This implies a bound on the expected value of $\Gamma$, which completes the proof of Theorem~\ref{thm:beta}.

\begin{lemma}
\label{lem:finalbound}
	For any $t \geq 0$, $\mathbb{E} [ \Gamma(t) ] \leq \frac{4C}{\hat{\alpha} \epsilon} n.$
\end{lemma}
\begin{proof}
By induction. This holds for $t = 0$ by a direct computation: see Lemma~\ref{lem:gammabound} for the argument. 
Then, we have 
\begin{eqnarray*}
	\mathbb{E}[ \Gamma(t + 1) ] =  \E [ \E[ \Gamma(t + 1) | \Gamma(t) ] ]  \leq \E \left[ \left( 1 - \frac{\hat{\alpha} \epsilon}{4 n } \right) \Gamma(t) + C \right] \leq   \frac{4C}{\hat{\alpha} \epsilon} n \left( 1 - \frac{\hat{\alpha}\epsilon}{4n} \right) + C \leq  \frac{4C}{\hat{\alpha} \epsilon} n. 
\end{eqnarray*} 
\end{proof}

\subsection{Guarantees on Max Rank}
\label{sec:maxrank}
\noindent We can use the characterization of the exponential process to prove the following: 
\begin{lemma}
	\label{lem:Emaxlabel}
	If $w_{\max}(t)$ is the maximum bin weight at time $t$ and $w_{\min}(t)$ is the minimum, then 
	\begin{eqnarray} 
	\E \left[ w_{\max}(t) - w_{\min}(t)  \right] = O \left( \frac{1}{\alpha} n ( \log n + \log C) \right). 
	\end{eqnarray}
\end{lemma}
\begin{proof} 
	Let $x_{\max} (t) =w_{\max}(t) / n$ and $x_{\min} (t) = w_{\min} (t) / n$.
	By definition, we have $\exp (\alpha (x_{\max} (t) - \mu(t) ) ) \leq \Gamma (t)$ and $\exp ( \alpha (\mu(t) - x_{\min} (t) ) ) \leq \Gamma (t)$.
	Therefore $\alpha (x_{\max} (t) - x_{\min} (t)) \leq 2 \log \Gamma (t)$.
	Thus, we have 
	\begin{align*}
	\E [\alpha (x_{\max} (t) - x_{\min} (t))] &\leq 2 \E [\log \Gamma (t)] \stackrel{(a)}{\leq} 2 \log (\E [\Gamma (t)]) \stackrel{(b)}{=} O(\log n + \log C) \; ,
	\end{align*}
	where (a) follows from Jensen's inequality and (b) follows from Theorem \ref{thm:beta}.
	Simplifying yields the desired claim.
\end{proof}
\noindent In Appendix \ref{sec:app-proofs-maxrank} we show that this implies the following theorem:
\begin{theorem}
\label{thm:Emaxrank}
	For all $t \geq 0$, we have $\E [\maxrank (t)] = O \left( \frac{1}{\alpha} n ( \log n + \log C) \right)$.
\end{theorem}

\noindent By plugging in constants as in (\ref{eq:param-delta}) and (\ref{eq:param-beta}), we get:
\begin{corollary}
\label{cor:max}
	For all $t \geq 0$, and any $\beta = \Omega (\gamma)$, we have $\E [\maxrank (t)] = O \left( \frac{1}{\beta} n ( \log n + \log 1 / \beta) \right)$.
\end{corollary}

%
%
%

\subsection{Guarantees on Average Rank}
\label{sec:avgrank}
We now focus on at the rank cost paid in a step by the algorithm. 
Let $A = \log C / \alpha$.
For real values $s \geq 0$, we ``stripe" the bins according to their top value, denoting by $b_{>s} (t)$ the number of bins with $w_j (t) \geq (s + A) n + \mu $
at time $t$, and let $b_{<- s} (t)$ be the number of bins with $w_j (t) \leq \mu - (s + A) n $ at time $t$. 
For any bin $j$ and interval $I$, we also let $\ell_{j, I} (t)$ be the number of elements in $j$ at time $t$ with label in $I$.
Finally, let $p_{j, \mu}$ denote the PDF of $w_j (t)$ given $\mu_t$.

We discuss the proof strategy at a high level here, deferring the full proofs to Appendix \ref{sec:app-proofs-avgrank}.
First, we use the bounds on $\Gamma$ to obtain the following bounds on the quantities defined previously:

\begin{lemma}
\label{lem:b}
	For any time $t$, we have that $\E[b_{>s}] \leq n\exp(-\alpha s)$ and that $\E[b_{<- s}] \leq n \exp(- \alpha s)$. 
\end{lemma}
\begin{proof}
	Recall that $\Phi(t) =  \sum_{i = 1}^n \exp\left( \alpha \left( x_i(t) - \mu \right)\right)$ and that $\E\left[ \Phi(t) \right] \leq C n$. 
	By linearity of expectation, we have $$\E[ \Phi(t) ] \geq \E[ b_{>s} \exp (\alpha( s + A))] = C \exp(\alpha s) \cdot \E[ b_{> s} ],$$
	\noindent which implies the claim. The converse claim follows from the bound on $\Psi(t)$. 
\end{proof}

This lemma gives us strong bounds on the tail behavior of the $w_j$.
We show that this implies a bound on the average rank. We will need the following technical result.
\begin{replemma}{lem:poi1}
For any interval $I = [a, b]$ which may depend on $\mu$, we have $\E [\rank (b) - \rank (a)| \mu] \leq n (b - a + 1)$.
\end{replemma}

\noindent We then show a bound on the rank of $\mu$:
\begin{replemma}{lem:below-mu}
	 For all $t$, we have $E [\rank (\mu(t))] \leq O \left( (A + 1 / \alpha^2) n \right)$.
\end{replemma}

The proof of this lemma, deferred to Appendix \ref{sec:app-proofs-avgrank}, works as follows.
We divide up the interval $(-\infty, \mu]$ into infinitely many intervals of constant length.
We will count the number of $w_j (t)$ within each interval separately.
Within each such interval, we can give a crude upper bound the number of weights that could have ever been in that interval at any point in the execution.
We then use Lemma \ref{lem:b} to show that at the current time in the execution, the number of bins with elements in each interval decays exponentially as we take intervals which are further and further away from $\mu$.
Summing these values up gives the desired bound.
Finally, this implies:
\begin{theorem}
\label{thm:avgrank}
	For all $t$, we have $\E \left[ \frac{1}{n} \sum_{i = 1}^n \rank_j (t) \right] = O \left(A +  \frac{1}{\alpha^2} \right) n$.
\end{theorem}
\begin{proof}
By Lemma \ref{lem:below-mu} and Lemma \ref{lem:poi1}, we have that 
\begin{align*}
\E [\rank_j (t) | w_j (t) \leq \mu (t) + An] &\leq O \left( (A + 1 / \alpha^2) n \right) + (A + 1)n \\
& = O \left( (A + 1 / \alpha^2) n \right) \; .
\end{align*}
Thus, we have
\begin{align*}
\E \left[  \sum_{i = 1}^n \rank_j (t) \right] &= \E \left[ \sum_{w_j(t) \leq \mu + An} \E[ \rank_j (t) | \mu ] + \sum_{w_j (t) > \mu + An} \E[ \rank_j (t) | \mu ] \right] \\
&= \E \left[ \sum_{w_j(t) \leq \mu + An} \E[ \rank_j (t) | \mu ] + \sum_{j = 1}^n \sum_{w_j (t) > \mu + An} \E[ \rank_j (\mu + An) + \ell_{j, (\mu, w_j (t))} | \mu, w_j (t) ] \right] \\
&\stackrel{(a)}{\leq}   O \left( \left( A + \frac{1}{\alpha^2 } \right) n^2 \right) + \sum_{j = 1}^n \E \left[ \sum_{w_j (t) > \mu + An} \E[ \ell_{j, (\mu, w_j (t))} | \mu, w_j (t) ]  \right] \\
&\stackrel{(b)}{\leq} O \left( \left( A + \frac{1}{\alpha^2 } \right) n^2 \right) + \sum_{j = 1}^n \E \left[ \sum_{w_j (t) > \mu + An} (w_j(t) - \mu + 1)  \right] \\
&= O \left( \left( A + \frac{1}{\alpha^2 } \right) n^2 \right) + \sum_{j = 1}^n \sum_{k = 0}^\infty \E \left[ \sum_{w_j (t) \in [\mu + (A + k)n, \mu + (A + k + 1) n ]} (w_j(t) - \mu + 1)  \right] \\
&\leq O \left( \left( A + \frac{1}{\alpha^2 } \right) n^2 \right) + n \sum_{k = 0}^\infty (A + k + 2) n \E \left[ b_{>k} \right] \\
&\stackrel{(c)}{\leq}  O \left( \left( A + \frac{1}{\alpha^2 } \right) n^2 \right) + O\left( \frac{1}{\alpha^2} n^2 \right).
\end{align*}
where (a) follows by Lemma~\ref{lem:below-mu}, (b) follows by Lemma \ref{lem:poi1}, and (c) follows from Lemma \ref{lem:b}.
Thus $$\E \left[ \frac{1}{n} \sum_{i = 1}^n \rank_j (t) \right] = O \left(A +  \frac{1}{\alpha^2} \right) n$$ as claimed.
\end{proof}

We now show how these imply bounds for our removal processes.
The actual rank choice at time $t$ is always better than a uniform choice in expectation, since it uses power of two choices.
If we consider an $(1 + \beta)$ process, where we only do two choices with probability $\beta = \Omega(\gamma)$, we obtain the following, by setting parameters as in (\ref{eq:param-delta}) and (\ref{eq:param-beta}) : 
\begin{corollary}
	\label{cor:avg}
	For all $t$, if we let $\beta = \Omega(\gamma)$, and we let $r(t)$ denote the rank of the removed element at time $t$, then
	\[
	\E [r(t)] = O \left( \left( \frac{\log C}{\alpha} + \frac{1}{\alpha^2}\right) n \right) =O \left( \frac{n}{\beta^2} \right) \; .\]
\end{corollary}

\section{Experimental Results}
\label{appendix:tests}

\paragraph{Setup and Methodology.} 
We implemented a $(1 + \beta)$ priority queue based on the MultiQueue implementation from the priority queue benchmark framework of~\cite{Wimmer}, and benchmarked it against the original MultiQueue ($\beta = 1$), the Linden-Jonsson~\cite{linden2013skiplist} skiplist-based implementation, and the kLSM deterministic-relaxed data structure~\cite{klsm}, with a relaxation factor of $256$, which has been found to perform best. The MultiQueue uses efficient sequential priority queues from the $\mathsf{boost}$ library. 
Tests were performed on a recent Intel(R) Xeon(R) CPU E7-8890 (Haswell architecture), with 18 hardware threads, each running at 2.5GHz.   
The tests for mean rank returned use coherent timestamps to record the times when elements are returned at each thread. 
We use these in a post-processing step to count rank inversions. This methodology might not be $100\%$ accurate, since the use of timestamps might change the schedule; however, we believe results should be reasonably close to the true values. 

For the throughput experiments, we consider executions consisting of alternating insert and deleteMin operations, for $10$ seconds. 
Experiment outputs are averaged over $10$ trials. 
Removals on empty queues do not count towards throughput. Since we are interested in the regime where queues are never empty, we insert 10 million elements initially. Threads are pinned to cores, and memory allocation is affinitized. 
The single-source shortest paths benchmark is a version of Dijkstra's algorithm, running on a weighted, directed California road network graph.

\paragraph{Results.} Figure~\ref{fig:throughput} illustrates the throughput differential between the various implementations, As previously stated, the MultiQueue variants are superior to other implementations (except at very low thread counts). 
Of note, the variants with $\beta < 1$ improve on the standard implementation by up to $20\%$. 
Since throughput figures are not conclusive in isolation, we also benchmarked the average rank cost in Figure~\ref{fig:mean}. (The y axis is logarithmic.) 
Note that the increase in average cost due to the further $\beta$ relaxation is relatively limited. 
Results are conformant with our analysis for $\beta \geq 0.5$. The apparent inflection point at around $\beta = 0.5$ could be explained by the $\epsilon \geq \delta$  bias assumptions breaking down after this point, or by non-trivial correlations in the actual execution which mean that our analysis no longer applies. 

Finally, Figure~\ref{fig:sssp} gives a running times for a single-source shortest path benchmark, using a parallel version of Dijkstra's algorithm. 
We note that the relaxed versions with $\beta <1$ can be superior in terms of running time to the version with $\beta = 1$, by up to $10\%$. The version with $\beta = 0$ (not shown) is the fastest at low thread counts, but then loses performance at  thread counts $\geq 8$, probably because of excessive relaxation.

\begin{figure}
    \centering
    \begin{minipage}{0.47\textwidth}
        \centering
        \includegraphics[width=\textwidth]{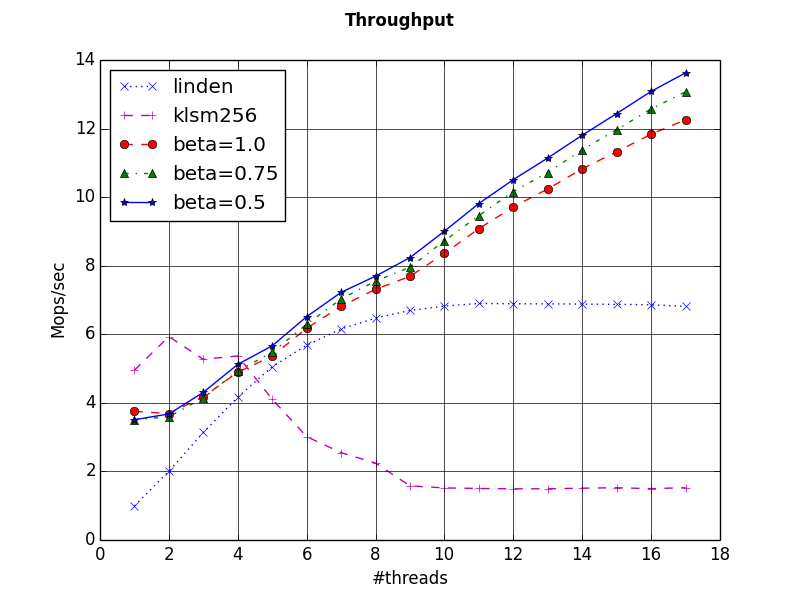} 
        \caption{Throughput comparison for  the $(1 + \beta)$ priority queue with $\beta = 0.5$ and $0.75$, versus the original MultiQueues, the Linden-Jonsson implementation, and kLSM. Higher is better.}
        \label{fig:throughput}
    \end{minipage}\hfill
    \begin{minipage}{0.47\textwidth}
        \centering
        \includegraphics[width=\textwidth]{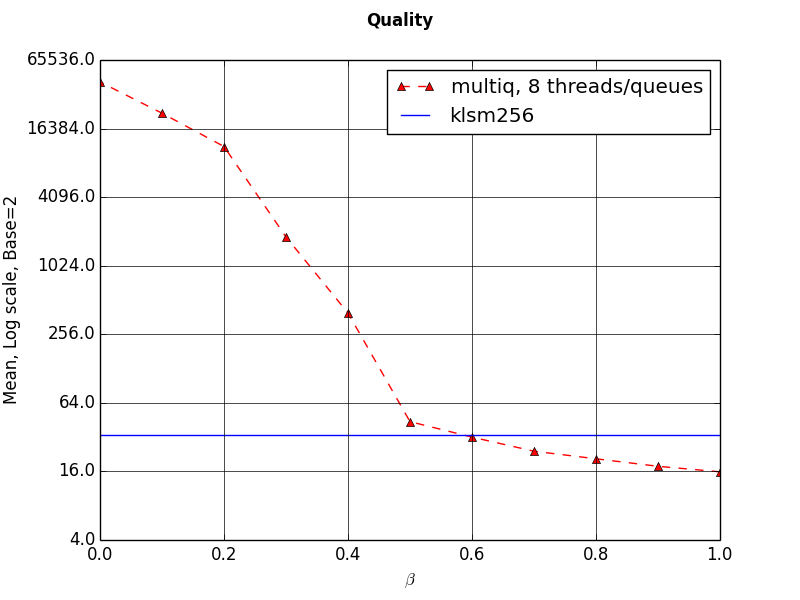} 
        \caption{Mean rank returned (log scale) for the $(1 + \beta)$ priority queue, for various values of $\beta$ on $8$ queues and $8$ threads. Lower is better.}
        \label{fig:mean}
    \end{minipage}
\end{figure}

\begin{figure}
    \centering
        \includegraphics[width=0.5\textwidth]{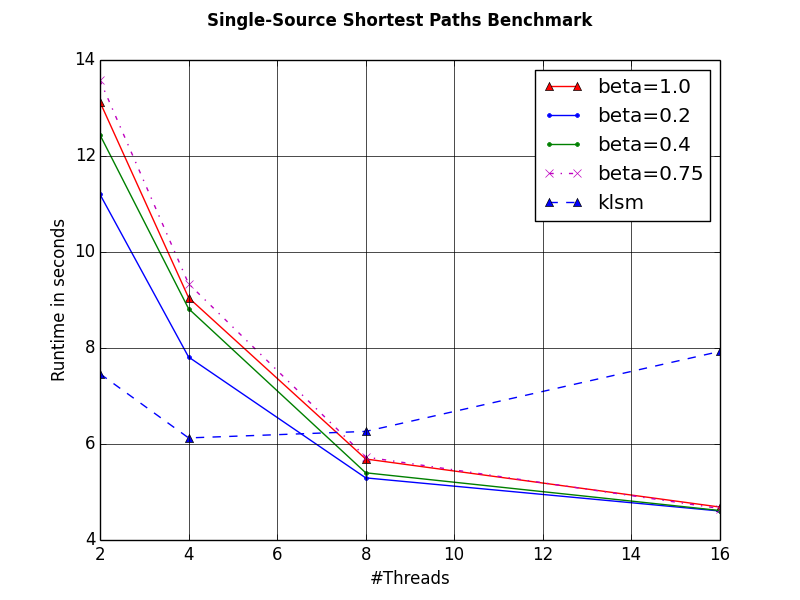} 
        \caption{Running times for single-source shortest path benchmark, using various versions of the priority queue, and kLSM. Lower is better.}
        \label{fig:sssp}
\end{figure}

\section{Discussion and Future Work}

We have provided tight rank guarantees for a practically-inspired randomized priority scheduling process. Moreover, we showed that this strategy is robust in terms of its bias and randomness requirements.  
Intuitively, our results show that, given biased random insertions into the queues, the preference towards lower ranks provided by the two-choice process is enough to give strong linear bounds on the average rank removed. We extended our analysis to a practical algorithm which improves on the state-of-the-art MultiQueues in Section~\ref{appendix:tests}. 

\paragraph{Tightness.} 
The bounds we provide for the two-choice process are asymptotically tight. To see that the $O(n)$ bound is tight, it suffices to notice that, even in a best-case scenario, the rank of the $k$th least expensive queue is at least $k$. Hence, the process has to have expected rank cost $\Omega( n )$. The tightness argument for $\Theta( n \log n  )$ expected worst-case cost is more complex. 
In particular, it is known~\cite[Example 2]{PTW15} that the gap between the most loaded and average bin load in a weighted balls-into-bins process with weights coming from an exponential distribution of mean $1$ is $\Theta( \log n )$ in expectation. That is, there exists a queue $j$ from which we have removed $\Theta( \log n )$ fewer times than the average.  
We can extend the argument in Section~\ref{sec:avgrank} to prove that there exist $\Theta(n)$ queues which have $\Theta( \log n )$ elements of higher label than the top element of queue $j$. This implies that the \emph{rank cost} of queue $j$ is $\Theta( n \log n )$, as claimed. 
We conjecture that the dependency in $\beta$ for the expected rank bound on the $(1 + \beta)$ process can be improved to \emph{linear}.

\paragraph{Relation to Concurrent Processes.} 
An important question is whether we can show that the practical \emph{concurrent} implementations  provide similar bounds. 
This claim holds if we assume that, in the concurrent implementation, the steps comparing the two top elements of the randomly chosen queues and removing the better top element are all performed \emph{atomically}, e.g. through hardware transactional memory. In this case, we can show that the probability distributions of ranks returned by the concurrent and sequential processes are identical, and hence our guarantees apply. We formalize the required property, called \emph{distributional linearizability}, in Appendix~\ref{sec:sequential-concurrent}. 

This property does not appear to hold in general for fine-grained lock-based or lock-free implementations. Appendix~\ref{sec:sequential-concurrent} provides a detailed discussion of the difficulties and limitations of extending this approach to a concurrent setting. In particular, due to complex correlations implied by concurrent execution, our bounds might not hold for some locking strategies for subtle reasons, and it appears highly non-trivial to rigorously extend them to any particular fine-grained strategy. 

In future work, we plan to examine whether  exist \emph{relaxations} of distributional linearizability satisfied by real implementations, which would allow our bounds to be extended to concurrent implementations without additional assumptions. The results in Section~\ref{appendix:tests} suggest that existing implementations do satisfy strong rank guarantees.  

\paragraph{Future Work.} 
An extremely intriguing question regards the impact of using such relaxed priority queue structures in the context of existing parallel algorithms, e.g. for shortest-paths computation. 
In particular, it would be interesting to bound the amount of extra work caused by relaxation versus the benefits of parallelism.

Another interesting question in terms of future work is whether our results can also be extended to processes on \emph{graphs}. 
More precisely, imagine we are given a strongly connected, undirected graph with $n$ vertices. 
Insertions of increasing integer labels are performed at uniformly random nodes. 
In every removal step, we pick a random edge, remove the higher priority label among its endpoints, and pay the rank of this label among all labels. 
It is not hard to see that, for graph families with good expansion, our analytic framework can be extended to this process as well. 
We plan to explore tight bounds for such graphical processes in future work.


\bibliographystyle{plain}
\bibliography{bibliography}

\appendix

\section*{Appendix}

The Appendix proceeds as follows. 
We give a reduction between classic power-of-two-choices processes and our sequential process, under round-robin insertions, in Section~\ref{sec:reduction-rr}.  We prove that the single random choice process diverges in terms of rank cost in Section~\ref{sec:diverge}. 
We cover the definition of distributional linearizability and relation to concurrent processes in Section~\ref{sec:sequential-concurrent}. 
We give complete proofs in Sections~\ref{sec:full},~\ref{sec:app-proofs-maxrank}, and~\ref{sec:app-proofs-avgrank}.

\section{Reduction to Two-Choices Process for Round-Robin Insertions}
\label{sec:reduction-rr}

\paragraph{Process Description.} For clarity, we re-define the process for round-robin insertions: 
  we are given $n$ queues, into which consecutively labeled elements are inserted in round-robin order: at the $t$th insertion step, we insert the element with label $t$  into bin $t \pmod n$. 
To remove an element, we pick two queues at random, and remove the element of lower label on top of either queue. 
We again measure the cost of a removal as the rank of the removed label among labels still present in any of the queues.

\paragraph{Reduction.} We will associate a ``virtual" bin $V_i$ with each queue $i$. Whenever an element is removed from the queue $i$, it gets immediately placed into the corresponding virtual bin. 
Notice that in every step we are removing the element of minimum label, which, by round-robin insertion, corresponds to the queue having been removed from less times than the other choice. 
Alternatively, we are inserting the removed element into the \emph{less loaded} virtual bin, out of two random choices (implicitly breaking ties by bin ID). 
This is the classic definition of the two-choices process~\cite{ABKU}, which has been analyzed for the long-lived case in e.g.~\cite{Berenbrink00}. 
One can prove bounds on the rank of elements removed using the existing analyses of this process, e.g.~\cite{Berenbrink00, PTW15}. 

\paragraph{The Reduction Breaks for Random Insertions.} Notice that the critical step in the above argument is the fact that, out of two random choices $i$ and $j$, 
the less loaded virtual bin corresponds to the queue which has lower label on top. This does not hold in the random insertion case; for instance, it is possible for a bin to get two elements with consecutive labels, which breaks the above property. 
Upon closer examination, we can observe that in the random insertion case, the correlation between the event that queue $i$ has lower label than queue $j$ and the event that queue $i$ has been removed from less times than queue $j$ exists, but is too weak to allow a direct reduction. 

\section{The Single Choice Process Diverges}
\label{sec:diverge}

In this section, we prove the following claim. 

\begin{theorem}
	\label{thm:diverge}
	The expected maximum rank guarantee of the process which inserts and removes from uniform random queues in each time step $t$ evolves as $\Omega( \sqrt{t n \log n} )$, for  $t = \Omega (n \log n )$.  
\end{theorem}
\begin{proof}
	First, notice that we can apply the reduction in Section~\ref{sec:reduction-rr} to obtain that the maximum number of elements removed from a queue is the same as the maximum load of a bin in the classic long-lived random-removal process. (This holds irrespective of the insertion process, since we are performing uniform random removals.) The maximum load of a bin in the classic process is known to be $\Theta\left(t / n + \sqrt{\frac{t}{n} \log n} \right)$, with high probability, for $t > n \log n$~\cite{PTW15}. 
	
	Let the queue with the most removals up to time $t$ be $Q$. For any other queue $Q_i$ with less removals, let $\ell( Q_i )$ be the difference between the number of elements removed from $Q$ and the number of elements removed from $Q_i$. 
	Notice that the expected number of elements removed from a queue up to $t$ is $t / n$. Hence, it is easy to prove that, with probability at least $1 / 2$, there are at least $n / 4$ queues which have had at most $t / n$ elements removed. Hence, for any such queue $Q_i$, $\ell( Q_i ) \geq \sqrt{\frac{t}{n} \log n}$. Hence, in expectation, $\sum_i \ell(Q_i) = \Omega (\sqrt{tn \log n})$. 
	
	We say that an element is \emph{light} if its label is \emph{smaller} than the label of the the top element of $Q$. 
	By symmetry, we know that each element counted in $\sum_i \ell(Q_i)$ is light with probability $\geq 1/2$. 
	Putting everything together, we get that, in expectation, we have $\Omega( \sqrt{tn \log n } )$ elements which are \emph{light}, i.e. of smaller label than the top label of $Q$. This implies that the expected rank of $Q$ is $\Omega( \sqrt{tn \log n } )$, as claimed. 

\end{proof}

\section{Relationship to Concurrent Processes}
\label{sec:sequential-concurrent}

On first glance, it might seem that a simple lock-based strategy should be linearizable to the sequential process we define in Section~\ref{sec:def}. For example, we could lock both queues that are to be examined (locking in order of their index to avoid deadlock and restarting the operation on failure) and declare the linearization point to be the point at which the second lock is grabbed. While this does linearize to \emph{some} relaxed sequential process, it turns out that our upper bounds fail to hold for subtle reasons when concurrency is introduced.

Consider the extreme execution in which process $0$ grabs the locks on some two queues, say $Q_i,Q_j$, and then hangs for a long time. Meanwhile, all the other processes complete many operations while process $0$ holds these locks, and all operations will have to retry if they try to grab locks on $Q_i$ or $Q_j$. In this case, many delete operations will be performed, none of which can delete from $Q_i$ or $Q_j$. Such an execution could produce arbitrarily bad rank errors. We can formalize the property we would like a correct concurrent strategy to have as \emph{distributional linearizability}:

\begin{definition}[Distributional Linearizability]
A randomized data structure $Q$  is \emph{distributionally linearizable} to a sequential data structure $S$ if for any parallel asynchronous execution, there exists a linearization of the operations of $Q$ to operations of $S$ so that the outputs of each operation of $Q$ are distributionally equivalent to those $S$.
\end{definition}

The simple locking strategy fails to be distributionally linearizable due to, e.g., the counter-example execution above. We would like to ask if there is a distributionally linearizable strategy, yet there appears to be an inherent limitation. Consider three processes $i, j, k$ performing \del{} operations in lock-step. No matter what strategy is used, if some two processes, say $i,j$, try to delete from the same queue $Q$, at least one, say $j$, will necessarily be delayed. As a result $i$ and $k$ will finish their operations while $j$ takes longer, causing $i$ and $k$ to be linearized before $j$ (this can be forced by e.g. an insertion to $Q$ while $j$ is delayed). As a result, an additional constraint has been introduced which requires that the first two \del{} operations to complete in this execution cannot have deleted from the same queue, a constraint which is not present for the sequential process.

In this way, the distribution of outputs of any concurrent process is affected by subtle timing issues in ways which seem hard to circumvent. We conjecture that there is no concurrent algorithm which is distributionally linearizable to the sequential process. That said, it may be possible that one can produce a simple strategy satisfying a weaker condition than distributional linearizability such as an analogous variant of sequential consistency or quiescent consistency.

\section{Omitted Proofs from Section \ref{sec:potential}}
\label{sec:full}
Here, we give a complete version of the potential argument. 

\begin{lemma}
\label{lem:psibound}
	$$\E \left[ \Delta \Psi \,|\, y(t) \right] \leq \sum_{i = 1}^n \hat{\alpha}\left(  (1 + \delta) \frac{1}{n} - p_i  \right) \exp\left( - \alpha y_i(t) \right).$$
\end{lemma}
\begin{proof}
\paragraph{Case 1:} If the bin is chosen, then the change is:

\begin{eqnarray*}
	\Delta \Psi_i  &=& \Psi_i(t+ 1) - \Psi_i(t) = \exp\left( - \frac{\alpha}{n} ( y_i(t) + \Delta_i \left(1 - \frac{1}{n} \right) ) \right) - \exp\left( -  \frac{\alpha}{n} y_i(t)  \right) \\ 
	&= & \exp\left( -\frac{\alpha}{n} y_i(t) \right) \left( \exp \left( - \frac{\alpha}{n} \Delta_i \left(1 - \frac{1}{n} \right) \right) - 1 \right)
\end{eqnarray*}

Taking expectations with respect to the random choices made on insertion, we have 
\begin{eqnarray*}
	\E_i \left[  \exp \left( - \Delta_i \frac{\alpha}{n}  \left(1 - \frac{1}{n} \right) \right)  \right] &= & \frac{1}{ 1 +  \frac{\alpha }{\pi_i n} \left( 1 - \frac{1}{n}  \right) } \\
	&=& 1 -  \frac{\alpha }{\pi_i n}\left( 1 - \frac{1}{n}  \right)  + \left( \frac{\alpha }{\pi_i n} \left( 1 - \frac{1}{n}  \right) \right)^2 - \ldots \\
  &\leq & 1 -  \frac{\alpha }{\pi_i n} \left( 1 - \frac{1}{n}  \right)  + \left( \frac{\alpha }{ \pi_i n} \left( 1 - \frac{1}{n}  \right)  \right)^2.
\end{eqnarray*}

\paragraph{Case 2:} If some other bin $j \neq i$ is chosen, then the change is:
\begin{eqnarray*}
	\Delta \Psi_i = & \exp\left( - \frac{\alpha}{n} ( y_i(t) - \Delta_j \frac{1}{n}  ) \right) - \exp\left( - \frac{\alpha}{n} y_i(t)  \right) = \\ 
	= & \exp\left( - \frac{\alpha}{n} y_i(t) \right) \left( \exp \left( \frac{\alpha \Delta_j}{n^2}  \right)  - 1 \right).
\end{eqnarray*}

Again taking expectations with respect to the random choices made on insertion, we have 
\begin{eqnarray*}
	\E_i \left[  \exp \left( \Delta_i \frac{\alpha}{n^2}  \right)  \right] = & \frac{ \pi_j  }{ \pi_j -  \frac{\alpha}{n^2}  }  = 
	\frac{ 1 }{ 1  -  \frac{\alpha}{\pi_j n^2}   } = 1 + \frac{\alpha}{\pi_j n^2} + \left( \frac{\alpha}{\pi_j n^2} \right)^2 + \ldots \leq  1 + \frac{\alpha}{\pi_j n^2} + c \left( \frac{\alpha}{\pi_j n^2} \right)^2.
\end{eqnarray*}

Therefore, we have that 

\begin{eqnarray*}
	\E \left[ \Delta \Psi_i \right]  &\leq & \left(  1 - \frac{\alpha }{n \pi_i}  \left( 1 - \frac{1}{n} \right) + \left(  \frac{\alpha }{n \pi_i } (1 - \frac{1}{n} ) \right)^2 \right) p_i  - 1 
	+ \sum_{j \neq i} \left(  1 + \frac{\alpha}{n^2 \pi_j }  + c \left( \frac{\alpha}{n^2 \pi_j }   \right)^2 \right) p_j  \\ 
	&\leq & 
	\left(  - \frac{\alpha}{n \pi_i}  + c \left( \frac{\alpha}{n \pi_i} \right)^2  \right)p_i + \sum_{i = 1}^n \left( \frac{\alpha}{n^2\pi_j} + c \left(\frac{\alpha}{n^2 \pi_j}  \right)^2 	 \right) p_j  \\ 
	&\leq & - \alpha p_i \left(  1 - \gamma  - c \alpha (1 + \gamma)^2 \right) p_i + \frac{\alpha}{n} \left( 1 + \gamma + \alpha c (1 + \gamma)^2 \right)  \\
	&\leq & \hat{\alpha} \left(  (1 + \delta) \frac{1}{n} - p_i  \right).
\end{eqnarray*}
\end{proof}

\begin{lemma}
\label{lem:phigood}
If $y_{n / 4} \leq 0$, then we have that 
\begin{eqnarray*}
\E \left[  \Phi(t + 1) \,|\, y(t) \right] 
\leq \left( 1 -   \frac{\hat{\alpha} \epsilon}{3n} \right) \Phi(t) + 1.
\end{eqnarray*}
\end{lemma}
\begin{proof}

We start from the inequality
\begin{equation}
\E \left[ \Delta \Phi \,|\, y(t) \right] \leq  \sum_{i = 1}^n \hat{\alpha}\left(  (1 + \delta) p_i - \frac{1}{n}  \right) \exp\left( \alpha y_i(t) \right) \leq \hat{\alpha} (1 + \delta)  \sum_{i = 1}^n p_i \exp\left( \alpha y_i(t) \right) - \frac{\hat{\alpha}}{n} \Phi(t).
\end{equation}

\noindent We will now focus on bounding the first term (without the constants). We can rewrite it as: 
\begin{equation}
\sum_{i = 1}^n p_i \exp\left( \alpha y_i(t) \right) = \sum_{i = 1}^{n / 4 - 1} p_i \exp\left( \alpha y_i(t) \right) + \sum_{i = n / 4}^{n} p_i \exp\left( \alpha y_i(t) \right).
\end{equation}

\noindent Since $y_{n / 4} \leq 0$, the first term is upper bounded by $1$. For the second term, notice that 
\begin{equation}
\sum_{i = n / 4}^{n} p_i \exp\left( \alpha y_i(t) \right) = \sum_{j = 1}^{3 n / 4} p_{n - j + 1} \exp{\left( \alpha y_{n - j + 1} (t) \right)}.
\end{equation}

\noindent The $p$ terms are non-decreasing in $j$, while the $y$ terms are non-increasing in $j$. Further, note that 
\begin{equation}
	\sum_{j = 1}^{3n/4} \exp{\left( \alpha y_{n - j + 1} (t) \right)} \leq \Phi. 
\end{equation}

\noindent The whole sum is maximized when these non-increasing terms are equal. 
We therefore are looking to bound 
\begin{equation}
\sum_{i = n / 4}^{n} p_i \exp\left( \alpha y_i(t) \right) \leq \frac{4 \Phi}{3n} \sum_{i = n/4}^{n}  p_i.
\end{equation}

\noindent Notice that 
\begin{equation}
 \sum_{i = n/4}^{n}  p_i = 1 - \sum_{i = 1}^{n  / 4 - 1} p_i \leq 1 - \left(\frac{1}{4} + \epsilon \right) = 
 \frac{3}{4} - \epsilon.   
\end{equation}

\noindent Hence we have that 
\begin{eqnarray*}
\E \left[ \Delta \Phi \,|\, y(t) \right] &\leq& \hat{\alpha}  \left[ (1 + \delta) \sum_{i = 1}^n p_i \exp\left( \alpha y_i(t) \right) - \frac{1}{n} \Phi(t) \right] \\
&\leq& \hat{\alpha}  \left[ (1 + \delta) \frac{4 \Phi}{3n} \left( \frac{3}{4} - \epsilon \right) - \frac{1}{n} \Phi(t) +  ( 1 + \delta ) \right]  \\ 
&\leq& \hat{\alpha}  \left[  \frac{1 + \delta}{n}  \left( 1 - \frac{4 \epsilon}{3} \right) \Phi(t) - \frac{1}{n} \Phi(t) + (1 + \delta) \right] \\
&=& \hat{\alpha} \left[ \frac{ \Phi(t) }{n} \left[  \left( 1 - \frac{4\epsilon}{3} \right) (1 + \delta)  - {1} \right]   + (1 + \delta)    \right]  \\ 
&\leq &- \hat{\alpha}  \left[  \frac{4\epsilon}{3} (1 + \delta)  - {\delta} \right] \frac{ \Phi(t) }{n}   + 1 \leq - \hat{\alpha} (1 + \delta)  \frac{\epsilon}{3}  \frac{ \Phi(t) }{n}   + 1,
\end{eqnarray*}

\noindent where in the last step we have used the fact that $\delta \leq \epsilon$. Hence, we get that 
\begin{eqnarray*}
\E \left[  \Phi(t + 1) \,|\, y(t) \right] 
\leq \left( 1 -   \frac{\hat{\alpha} \epsilon}{3n} \right) \Phi(t) + 1.
\end{eqnarray*}

\end{proof}

\begin{lemma}
\label{lem:psigood}
If $y_{3n / 4} \geq 0$, then we have that 
\begin{eqnarray*}
\E \left[ \Psi (t + 1) \,|\, y(t) \right] \leq &  \left(   1 -  \frac{\hat{\alpha} \epsilon}{3n}  \right) \Psi + 1, 
\end{eqnarray*}
\end{lemma}
\begin{proof}
We start from the inequality
\begin{eqnarray*}
\E \left[ \Delta \Psi \,|\, y(t) \right] \leq & 
\sum_{i = 1}^n \hat{\alpha}\left(  (1 + \delta) \frac{1}{n} - p_i  \right) \exp\left( - \alpha y_i(t) \right)  \\ 
\leq & \sum_{i = 1}^{3n / 4} \hat{\alpha}\left(  (1 + \delta) \frac{1}{n} - p_i  \right) \exp\left( - \alpha y_i(t) \right) + 
\sum_{i = 3 n / 4 + 1}^{n} \hat{\alpha}\left(  (1 + \delta) \frac{1}{n} - p_i  \right) \exp\left( - \alpha y_i(t) \right) . 
\end{eqnarray*}

\noindent The second term is upper bounded as 

\begin{eqnarray*}
\sum_{i = 3 n / 4 + 1}^{n} \hat{\alpha}\left(  (1 + \delta) \frac{1}{n} - p_i  \right) \leq 
\hat{\alpha} (1 + \delta ) / 4  - \hat{\alpha} \sum_{i = 3 n / 4 + 1}^{n} p_i \leq 1.
\end{eqnarray*}

\noindent We therefore now want to bound
\begin{eqnarray*}
\sum_{i = 1}^{3n / 4} \hat{\alpha}\left(  (1 + \delta) \frac{1}{n} - p_i  \right) \exp\left( - \alpha y_i(t) \right).
\end{eqnarray*}

\noindent We again notice that the first factor is non-decreasing, whereas the second one is non-increasing. 
Hence, the sum is maximized when the non-decreasing terms are equal.
At the same time, we have that  
\begin{eqnarray*}
\sum_{i = 1}^{3n/4} \exp\left( - \alpha y_i(t) \right) \leq \Psi. 
\end{eqnarray*}

Hence, it holds that 
\begin{eqnarray*}
\sum_{i = 1}^{3n / 4} \hat{\alpha}\left(  (1 + \delta) \frac{1}{n} - p_i  \right) \exp\left( - \alpha y_i(t) \right) \leq \frac{4 \Psi}{3n} \sum_{i = 1}^{3n / 4} \hat{\alpha}\left(  (1 + \delta) \frac{1}{n} - p_i  \right) = \hat{\alpha} \frac{1 + \delta}{n} \Psi  - \hat{\alpha} \frac{4\Psi}{3n} \sum_{i = 1}^{3n / 4} p_i  \\
 = \hat{\alpha} \frac{1 + \delta}{n} \Psi  - \hat{\alpha} \frac{1}{n} \Psi \left( 1 + \frac{4 \epsilon}{3}  \right) 
= \frac{\hat{\alpha}}{n} \Psi \left(  \left(1 + \delta\right) - \left(1 + \frac{4 \epsilon}{3} \right) \right) \leq  - \frac{\hat{\alpha}}{n}     \frac{\epsilon}{3} \Psi,
\end{eqnarray*}

\noindent where in the last step we have used the fact that $\delta \leq \epsilon$. Therefore, we get that 
\begin{eqnarray*}
\E \left[ \Psi (t + 1) \,|\, y(t) \right] \leq &  \left(   1 -  \frac{\hat{\alpha} \epsilon}{3n}  \right) \Psi + 1, 
\end{eqnarray*}

\noindent as claimed. 

\end{proof}

\begin{lemma}
\label{lem:phibad}
Given $\epsilon \in ( 0, 1 )$ as defined, assume that $y_{n / 4} > 0$ at $t$, and that $\mathbb{E}[ \Delta \Phi ] \geq - \frac{\epsilon \hat{ \alpha}}{ 3n } \Phi(t)$. 
 Then either $\Phi < \frac{\epsilon}{4} \Psi$, or $\Gamma = O( n )$.   
\end{lemma}
\begin{proof}

Fix $\lambda  = 2 / 3 - 1 / 54$  for the rest of the proof. 
We can split the inequality in Lemma~\ref{lem:phibound} as follows:
\begin{eqnarray}
\E \left[ \Delta \Phi \,|\, y(t) \right] \leq \sum_{i = 1}^{\lambda n} \hat{\alpha}\left(  (1 + \delta) p_i - \frac{1}{n}  \right) \exp\left( \alpha y_i(t) \right) + 
\sum_{i = \lambda n + 1}^{n} \hat{\alpha}\left(  (1 + \delta) p_i - \frac{1}{n}  \right) \exp\left( \alpha y_i(t) \right). 
\label{ineq:phisplit}
\end{eqnarray}

\noindent We bound each term separately. Since the probability terms are non-increasing and the exponential terms are non-decreasing, the first term is maximized when all $p_i$ terms are equal. Since these probabilities are at most $1$, we have
\begin{eqnarray}
\label{ineq:philess}
\sum_{i = 1}^{\lambda n} \hat{\alpha}\left(  (1 + \delta) p_i - \frac{1}{n}  \right) \exp\left( \alpha y_i(t) \right) \leq \frac{\hat{\alpha}}{n} \left( (1 +\delta) \frac{1}{\lambda} - 1 \right) \Phi_{\leq \lambda n}.
\end{eqnarray} 

\noindent The second term is maximized by noticing that the $p_i$ factors are non-increasing, and are thus dominated by their value at $\lambda n$. Noticing that we carefully picked $\lambda$ such that 
$$ p_{\lambda n} \leq \frac{1}{n} - \frac{4 \epsilon}{n}, $$
we obtain, using the assumed inequality $\delta \leq \epsilon$, that 
\begin{eqnarray}
\label{ineq:phimore}
	\sum_{i = \lambda n + 1}^{n} \hat{\alpha}\left(  (1 + \delta) p_i - \frac{1}{n}  \right) \exp\left( \alpha y_i(t) \right)  \leq \hat{\alpha} \left( \delta - 4 \epsilon  \right) \frac{\Phi_{>\lambda n}}{n} \leq - \frac{3 \epsilon \hat{\alpha}}{n} \Phi_{>\lambda n}. 
\end{eqnarray}

\noindent By the case assumption, we know that $
\E \left[ \Delta \Phi \,|\, y(t) \right] \geq - \frac{\hat{\alpha} \epsilon}{3n} \Phi(t).$
 Combining the bounds (\ref{ineq:phisplit}), (\ref{ineq:philess}), and (\ref{ineq:phimore}), this yields:
\begin{eqnarray*}
\frac{\hat{\alpha}}{n} \left( (1 +\delta) \frac{1}{\lambda} - 1 \right) \Phi_{\leq \lambda n}  - \frac{3 \epsilon \hat{\alpha}}{n} \Phi_{>\lambda n} 
 \geq - \frac{\hat{\alpha}\epsilon}{3n} \Phi(t). 
\end{eqnarray*}

\noindent Substituting $\Phi_{>\lambda n} =  \Phi - \Phi_{\leq \lambda n}$ yields:
%
%
\begin{eqnarray*}
 \left(  3 \epsilon - \epsilon / 3 \right) \Phi \leq \left(  (1 + \delta) \frac{1}{\lambda} - 1 + 3 \epsilon  \right) \Phi_{\leq \lambda n} . 
\end{eqnarray*}

\noindent For simplicity, we fix 
$
 C(\epsilon) =  \frac{  (1 + \delta) \frac{1}{\lambda} - 1 + 3 \epsilon }{ 3 \epsilon - \epsilon / 3 } = O\left( \frac{1}{\epsilon} \right), 
$
to obtain 
\begin{eqnarray}
\label{eq:phirel}
\Phi \leq C(\epsilon)  \Phi_{\leq \lambda n}. 
\end{eqnarray}

\noindent Let $B = \sum_{y_i > 0} y_i$. Since we are normalizing by the mean, it also holds that $B = \sum_{y_i < 0} (- y_i )$. 
Notice that 
\begin{eqnarray}
\label{eq:brel}
\Phi_{ \leq \lambda n} \stackrel{y_i \text{incr.}}{\leq} \lambda n \exp{\left(  {\alpha y_{\lambda n} } \right)} 
 \stackrel{y_i \text{incr.}}{\leq} \lambda n \exp{\left(  \frac{\alpha B }{ (1 - \lambda) n} \right)}.
\end{eqnarray}

\noindent We put inequalities~(\ref{eq:phirel}) and~(\ref{eq:brel}) together and get 
\begin{eqnarray}
\label{eq:phitwo}
	\Phi(t) \leq & \lambda n C(\epsilon)  \exp{\left(  \frac{\alpha B }{ (1 - \lambda) n} \right)} ,
\end{eqnarray}

\noindent Let us now lower bound the value of $\Psi$ under these conditions. Since $y_{n / 4} > 0,$ all the costs below average must be in the first quarter of $y$. 
We can apply Jensen's inequality to the first $n / 4$ terms of $\Psi$ to get that 
\begin{eqnarray*}
	\Psi \geq \sum_{i = 1}^{n / 4} \exp \left( - \alpha {y_i} \right) \geq 
	\frac{n}{4} \exp \left(  - {\alpha} \frac{\sum_{i = 1}^{n / 4} y_i }{  n / 4  }      \right).
\end{eqnarray*}

\noindent We now split the sum $\sum_{i = 1}^{n / 4} y_i$ into its positive part and its negative part. We know that the negative part is summing up to exactly $- B$, as it contains all the negative $y_i$'s and the total sum is $0$. The positive part can be of size at most $B / 4$,  since it is maximized when there are exactly $n - 1$ positive costs and they are all equal. 
Hence the sum of the first $n / 4$ elements is at least $ - 3B / 4$, which implies that the following bound holds: 
\begin{eqnarray}
\label{eq:psilb}
	\Psi \geq 
	\frac{n}{4} \exp{ \left(  - {\alpha} \frac{- 3B / 4 }{  n / 4  } \right)} \geq 
	\frac{n}{4} \exp \left(  {\alpha} \frac{ 3B }{ n } \right). 
\end{eqnarray}

\noindent If  $\Phi < \frac{\epsilon}{4} \Psi$, then there is nothing to prove. 
Otherwise, if $\Phi \geq  \frac{\epsilon}{4} \Psi$, we get from~(\ref{eq:psilb}) and~(\ref{eq:phitwo}) that 
\begin{eqnarray*}
	\frac{\epsilon}{4} \frac{n}{4} \exp \left(  {\alpha} \frac{ 3B }{ n } \right) \leq \frac{\epsilon}{4} \Psi  \leq  \Phi(t) \leq & \lambda n C(\epsilon)  \exp{\left(  \frac{\alpha B }{ (1 - \lambda) n} \right)} ,
\end{eqnarray*}
\noindent Therefore, we get that: 
\begin{eqnarray*}
	 \exp \left(  {\alpha} \frac{ B }{ n } \left(  3 - \frac{1}{1 - \lambda}  \right) \right) \leq \frac{4\lambda}{\epsilon} C(\epsilon) = O\left( \frac{1}{\epsilon^2} \right). 
\end{eqnarray*}

\noindent Using the mundane fact that $  3 - \frac{1}{1 - \lambda} = 3 / 19$, we get that 
%
\begin{eqnarray}
\label{Bbound}
	 \exp \left(   \frac{ {\alpha}B }{ n } \right) \leq O\left( \frac{1}{\epsilon^{14}}\right). 
\end{eqnarray}

\noindent To conclude, notice that (\ref{Bbound}) implies we can upper bound $\Gamma$  in this case as:
 \begin{eqnarray*}
 \Gamma  = \Phi + \Psi \leq & \frac{4 + \epsilon}{\epsilon} \Phi \leq \frac{4 + \epsilon}{\epsilon} \lambda n C(\epsilon)  \exp{\left(  \frac{\alpha B }{ (1 - \lambda) n} \right)}
 \leq  O\left( \frac{1}{\epsilon^{14/(1 - \lambda)}}\right) \frac{4 + \epsilon}{\epsilon}  C(\epsilon) \lambda n = O( \text{poly}\left( \frac{1}{\epsilon} \right ) n ). 
 \end{eqnarray*}

\end{proof}

\begin{lemma}
\label{lem:psibad}
Given $\epsilon$ as above, assume that $y_{\frac{3n}{4}} < 0$, and that $\mathbb{E}[ \Delta \Psi ] \geq -\frac{ \alpha \epsilon}{ 3 n } \Psi$. 
	Then either $\Psi > \frac{\epsilon}{4} \Phi$ or $\Gamma = O (n )$. 
\end{lemma}
\begin{proof}
	Fix $\lambda = \lambda_2 = 1 / 3 + 1 / 54$. 
	We start from the general bound on $\Delta \Psi$ from Lemma~\ref{lem:psibound}. 
	We had that 
	\begin{eqnarray*} 
		\E \left[ \Delta \Psi \,|\, y(t) \right] &\leq& \sum_{i = 1}^n \hat{\alpha}\left(  (1 + \delta) \frac{1}{n} - p_i  \right) \exp\left( - \alpha y_i(t) \right) \\ 
		&=& \sum_{i = 1}^{\lambda n - 1} \hat{\alpha}\left(  (1 + \delta) \frac{1}{n} - p_i  \right) \exp\left( - \alpha y_i(t) \right) + 
		\sum_{i = \lambda n}^{n} \hat{\alpha}\left(  (1 + \delta) \frac{1}{n} - p_i  \right) \exp\left( - \alpha y_i(t) \right) \\ 
		&\leq& \frac{\hat{\alpha}}{n} \left(  (1 + \delta) - (1 + 4 \epsilon) \right)  \Psi_{< \lambda n} +  \frac{\hat{\alpha}}{n} \left( 1 + \delta \right) \Psi_{\geq \lambda n} \\
		&\leq& - \frac{\hat{3 \alpha \epsilon}}{n} \Psi_{< \lambda n} +  \frac{\hat{\alpha}}{n} \left( 1 + \delta \right) \Psi_{\geq \lambda n}. 
	\end{eqnarray*}

\noindent Using the assumption, we have that
	\begin{eqnarray*} 
		-\frac{ \hat{\alpha} \epsilon}{ 3n } \Psi \leq \E \left[ \Delta \Psi \,|\, y(t) \right]
		\leq - \frac{3\hat{\alpha}\epsilon}{n}  \left( \Psi - \Psi_{\geq \lambda n} \right) +  \frac{\hat{\alpha}}{n} \left( 1 + \delta \right) \Psi_{\geq \lambda n}. 
	\end{eqnarray*}

We can re-write this as
	\begin{eqnarray*} 
		\Psi\leq \Psi_{\geq \lambda n} \frac{1 + \delta + 3 \epsilon }{  3 \epsilon - \epsilon / 3    } = C(\epsilon) \Psi_{\geq \lambda n}. 
	\end{eqnarray*}

\noindent However, we also have that
\begin{eqnarray*}
	\Psi_{\geq \lambda n} \leq ( 1 - \lambda ) n \exp{ \left(-\alpha y_{\lambda n} \right)} \leq 
	 ( 1 - \lambda ) n \exp{ \left( \frac{\alpha B}{n  \lambda } \right)}.
\end{eqnarray*}
 
 \noindent At the same time, since $y_{3n/4} < 0$, we have that
 \begin{eqnarray*} 
 \Phi \geq \frac{n}{4} \exp{ \left(  \frac{3\alpha B}{n}    \right) }.
 \end{eqnarray*}

\noindent  If $\Psi < \frac{\epsilon}{4} \Phi$, we can conclude. Let us examine the case where $\Psi \geq \frac{\epsilon}{4} \Phi$. 
Putting everything together, we get

\begin{eqnarray*}
	\frac{\epsilon}{4} \frac{n}{4} \exp\left( \frac{3\alpha B}{n} \right) \leq \frac{\epsilon}{4} \Phi \leq \Psi \leq C(\epsilon) \Psi_{\geq \lambda n} \leq \left(1 - \lambda\right) n C(\epsilon) \exp\left( \frac{\alpha B}{n \lambda}  \right).
\end{eqnarray*}

\noindent Alternatively, 
\begin{eqnarray*}
	\exp\left( \frac{\alpha B}{n} \left( 3 - \frac{1}{\lambda} \right) \right) \leq \frac{16}{ \epsilon} C(\epsilon) (1 - \lambda) = \frac{16}{ \epsilon}  (1 - \lambda) \frac{1 + \delta + 3 \epsilon }{  3 \epsilon - \epsilon / 3    } = O\left( \frac{1}{\epsilon^2} \right).
\end{eqnarray*}

\noindent Therefore, 
\begin{eqnarray*}
	\exp\left( \frac{\alpha B}{n} \right) \leq  O\left( \frac{1}{\epsilon^{13}} \right) .
\end{eqnarray*}

\noindent To complete the argument, we bound:

\begin{eqnarray*}
	\Gamma = \Psi + \Phi \leq  \left(1 + \frac{4}{\epsilon} \right) \Psi \leq \left(1 - \lambda\right) C(\epsilon) \exp\left( \frac{\alpha B}{n \lambda}\right) n 
	\leq   O \left( \frac{n}{\epsilon^{22}} \right). 
\end{eqnarray*}
\end{proof}

\section{Proof of Theorem \ref{thm:Emaxrank}}
\label{sec:app-proofs-maxrank}

Before we prove Theorem \ref{thm:Emaxrank} we need the following fact about exponential disitributions:

\begin{fact}
\label{fact:poisson}
Let $X_1, X_2, \ldots$ be independent and $X_i \sim \Exp (\lambda)$ for all $i$.
Let $Y_i = \sum_{k \leq i} X_k$.
Fix any interval $I \subseteq [0, \infty)$ of length $m$.
Then, $\# \{i: Y_i \in I \} \sim \Poi (m \lambda)$.
\end{fact}
 
 \begin{proof}[Proof of Theorem \ref{thm:Emaxrank}]
Let $I(t) = [ w_{\min} (t), w_{\max}(t)]$, and let $L_{j} (t)$ be the number of elements in bin $j$ in $I(t)$ at time $t$.
Then, by memorylessness and Fact \ref{fact:poisson}, for all bins $j$ except the bin $j_{\max}$ containing $w_{\max} (t)$, the number of labels in $I(t)$ in bin $j$ before any deletions occur is distributed as $\Poi (|I| / n)$.
In particular, the expected number of elements after $t$ rounds is bounded by $\E_{X \sim \Poi (|I| / n)} [X] = |I| / n$.
Moreover, for the bin containing $w_{\max} (t)$, by definition, at time $t$ it contains exactly one element in $I$, namely, $w_{\max} (t)$.
Hence, we have
\begin{align*}
\E [\maxrank (t)] &\leq \E_{I(t)} \left[ \E \left[ \sum_{j = 1}^n L_j (t) \middle| I(t)  \right] \right] \\
&\leq \E_{I(t)} \left[ 1 + (n - 1) |I| / n \right] \\
&= O \left( \frac{1}{\alpha} n ( \log n + \log C) \right) \; ,
\end{align*}
by Lemma \ref{lem:Emaxlabel}.
\end{proof}

\section{Omitted Proofs from Section \ref{sec:avgrank}}
\label{sec:app-proofs-avgrank}

We first require the following technical lemma:

\begin{lemma}
\label{lem:poi1}
For any interval $I = [a, b]$ which may depend on $\mu$, we have $\E [\rank (b) - \rank (a)| \mu] \leq n (b - a + 1)$.
\end{lemma}
\begin{proof}
We first show that if $I$ does not depend on $\mu$, then $\E [\rank (b) - \rank (a)| \mu] \leq n (b - a)$.
By Fact \ref{fact:poisson} we have
\begin{align*}
\E [\rank(b) - \rank (a) | \mu ] &\leq \E \left[ \sum_{j = 1}^n \ell_{j, [a, b]}  \middle| \mu \right] \\
&=\E \left[ \sum_{j = 1}^n \E_{X \sim \Poi (b - a)} [X] \middle| \mu \right] \\
&= n (b - a) \; .
\end{align*}
To conclude the proof, we now observe that $\mu$ depends on at most $n$ elements, namely, those on top of the queues, and that if we remove those elements, then the remaining elements behave just as above.
Thus, by conditioning on $\mu$, we increase the rank by at most an additional factor of $n$.
 \end{proof}

\begin{lemma} 
\label{lem:below-mu}
	 For all $t$, we have $E [\rank (\mu(t))] \leq O \left( (A + 1 / \alpha^2) n \right)$.
\end{lemma}
\begin{proof}
For any bin $j$, we have
\begin{align*}
\E [\ell_{j, (-\infty, \mu]} (t) | \mu ] &= \int_{-\infty}^\mu \E [\ell_{j, [x, \mu]} | \mu, w_j (t) = x ] \cdot p_{j, \mu} (x) dx \\
&\stackrel{(a)}{\leq}  \int_{-\infty}^\mu \E_{X \sim \Poi(\mu - x) / n} [1 + X | \mu, w_j (t) = x ] \cdot p_{j, \mu} (x) dx \\
&= \int_{-\infty}^\mu \frac{1}{n} (1 + \mu - x) \cdot p_{j, \mu} (x) dx = \frac{1}{n} \E [(1 + \mu - w_j(t)) \mathrm{1}_{w_j(t) \leq \mu} | \mu] \\
&\leq \frac{1}{n} \left( 1 +  \E [(\mu - w_j(t)) \mathrm{1}_{w_j(t) \leq \mu} | \mu] \right) \; .
\end{align*}
where (a) follows because after we've conditioned on the value of $w_j (t)$, the values of the remaining labels in bin $j$ is independent of $\mu$ and we can apply Fact \ref{fact:poisson}.
Therefore, we have
\begin{align*}
\E [\rank (\mu)] &\stackrel{(a)}{\leq} \E [\rank (\mu - An)] + (A + 1) n = \E \left[ \sum_{j = 1}^n \E [\ell_{j, (-\infty, \mu - An] (t)} | \mu] \right] + (A + 1) n \\
&=  \E \left[ \sum_{j = 1}^n \E [\ell_{j, (w_j(t), \mu - An] (t)} | \mu, w_j (t)] \right] + (A + 1) n \\
&\stackrel{(b)}{\leq} 1 + \frac{1}{n} \E_{\mu} \left[   \sum_{j = 1}^n (\mu - w_j(t)) \mathrm{1}_{w_j(t) \leq \mu - An} \right] + (A + 1) n \\
&\stackrel{(c)}{\leq} 1 + \frac{1}{n} \sum_{k = 1}^\infty \E_{\mu} \left[ \sum_{w_j (t) \in [\mu - (k + 1 + A) n, \mu - (k + A) n]} (\mu - w_j (t) + 1) \right] + (A + 1)n \\
&\leq 1 + \frac{1}{n} \sum_{k = 1}^\infty (k + 2) n \E_{\mu} \left[ b_{<- i} \right] + (A + 1)n \stackrel{(d)}{\leq} 1 + \sum_{k = 1}^\infty (k + 2) n  \exp (-\alpha k) + (A + 1)n \\
&\leq 1 +  \frac{e^\alpha}{(e^\alpha - 1)^2} n + (A + 1)n \;,
\end{align*}
where (a) follows from Lemma \ref{lem:poi1}, (b) follows from Fact \ref{fact:poisson}, (c) follows from Lemma \ref{lem:poi1}, and (d) follows from Lemma \ref{lem:b}. 
Notice that for $\alpha$ small we have $ \frac{e^\alpha}{(e^\alpha - 1)^2} = O(1 / \alpha^2)$.
Simplifying the above then yields the claimed bound.
\end{proof}

\subsection{Other Proofs}

\begin{lemma}
\label{lem:initial}
	We have that $\Gamma(0) = O(n)$. 
\end{lemma}
\begin{proof}
	In the initial state, the values $x_i(0)$ are independent exponentially distributed random variables with mean $1/(n\pi_i)$. Now,
\begin{eqnarray*}
\E[\Phi(0)] & = & \E\left[\sum\exp\left(\alpha\left(\frac{w_i(0)}{n} - \mu(0)\right)\right)\right] \\
& = & \sum\E\left[\exp\left(\alpha\left(x_i(0) - \frac{1}{n}\sum x_j\right)\right)\right] \\
& =  & \sum\E\left[\exp\left(\alpha x_i(0)\left(1-\frac{1}{n}\right)\right) \prod_{j\neq i}\exp\left(-\frac{\alpha}{n} x_j(0)   \right)\right] \\
& = &  \sum\E\left[\exp\left(\alpha x_i(0)\left(1-\frac{1}{n}\right)\right)\right]\E\left[ \prod_{j\neq i}\exp\left(-\frac{\alpha}{n} x_j(0)   \right)\right],
\end{eqnarray*}

where the last line follows from the (pairwise) independence of the $x_i$. All terms in the right hand side are now exactly moment generating functions of the $w_i$ evaluated at some point. Since the moment generating function of an exponential with parameter $\lambda$ evaluated at $t$ is well known to be $\lambda/(\lambda-t)$, we can compute:
\begin{eqnarray*}
\E[\Phi(0)] & = &  \sum \frac{\frac{1}{n\pi_i}}{\frac{1}{n\pi_i}-\alpha(1-\frac{1}{n})} \prod_{j\neq i}\frac{\frac{1}{n\pi_j}}{\frac{1}{n\pi_j}+\frac{\alpha}{n}} \\ 
& \leq & \sum \frac{1+\gamma}{1-\gamma-\alpha}\prod_{j\neq i} \frac{1}{1+\alpha\pi_j} \\
& \approx & \sum \frac{1+\gamma}{1-\gamma-\alpha}\left(1- \sum_{j\neq i} \alpha\pi_j\right) \\
& = & O(n).
\end{eqnarray*}

The argument for $\Psi(0)$ is symmetric, replacing $\alpha$ with $-\alpha$.

\end{proof}

%
%
%

\end{document}